\documentclass{IEEEtran} 
 
\usepackage{cite}
\usepackage{amsmath}
\usepackage{commath}
\usepackage[cmex10]{mathtools}
\usepackage{amssymb}
\usepackage{bm} 
\usepackage{amsthm}
\usepackage{color}
\usepackage{float} 
\usepackage[acronym]{glossaries}
\usepackage{multirow}
\usepackage{tabularx}
\usepackage{booktabs}
\usepackage{esvect} 
\usepackage{graphicx}

\usepackage{amsfonts}
\usepackage{amsmath}
\usepackage{amssymb}
\usepackage{cite}
\usepackage{dsfont}
\usepackage{latexsym}
\usepackage{subcaption}
\usepackage{times}
\usepackage{url}
\usepackage{verbatim}

\def\cA{{\mathcal{A}}}

\def\cE{{\mathcal{E}}}

\def\cH{{\mathcal{H}}}

\def\cJ{{\mathcal{J}}}

\def\cP{{\mathcal{P}}}

\def\cV{{\mathcal{V}}}

\def\ba{{\mathbf{a}}}
\def\bb{{\mathbf{b}}}
\def\bc{{\mathbf{c}}}

\def\bee{{\mathbf{e}}}

\def\bg{{\mathbf{g}}}

\def\bmm{{\mathbf{m}}}

\def\bp{{\mathbf{p}}}

\def\br{{\mathbf{r}}}
\def\bs{{\mathbf{s}}}

\def\bu{{\mathbf{u}}}
\def\bv{{\mathbf{v}}}
\def\bw{{\mathbf{w}}}

\def\bA{{\mathbf{A}}}

\def\bC{{\mathbf{C}}}

\def\bE{{\mathbf{E}}}

\def\bI{{\mathbf{I}}}

\def\bM{{\mathbf{M}}}

\def\bQ{{\mathbf{Q}}}
\def\bR{{\mathbf{R}}}
\def\bS{{\mathbf{S}}}
\def\bT{{\mathbf{T}}}
\def\bU{{\mathbf{U}}}
\def\bV{{\mathbf{V}}}

\def\bX{{\mathbf{X}}}

\DeclareMathOperator{\diag}{diag}
\DeclareMathOperator{\blkdiag}{blkdiag}
\DeclareMathOperator{\T}{\sf T}

\DeclareMathOperator*{\argmax}{arg\!\max}

\def\R{{\mathbb{R}}}
\def\C{{\mathbb{C}}}

\def\cj{{\sf j}}

\def\Re{{\text{Re}}}
\def\Im{{\text{Im}}}

\def\PD{{\sf PD}}
\def\Pcon{{\cP_{\sf con}}}

\newcommand{\cbE}{\boldsymbol{\mathcal{E}}}
\newcommand{\cbA}{\boldsymbol{\mathcal{A}}}

\def\Tff{\bT_{\sf ff}}
\def\Tarrff{\overline{\bT}_{\sf arr, ff}}

\def\Tbar{\overline{\bT}}
\def\Rbar{\overline{\bR}}
\def\mbar{\overline{\bmm}}
\def\Mbar{\overline{\bM}}
\def\Kbar{\overline{K}}

\def\Rarr{\overline{\bR}_{\sf arr}}
\def\Tarr{\overline{\bT}_{\sf arr}}

\def\bur{{\bu_{r}}}
\def\bue{{\bu_{\theta}}}
\def\bua{{\bu_{\phi}}}

\def\Edipr{{\cE_{{\sf dip}, r}}}
\def\Edipe{{\cE_{{\sf dip}, \theta}}}
\def\Edipa{{\cE_{{\sf dip}, \phi}}}

\def\bEdip{{\cbE_{\sf dip}}}
\def\Earr{{\cbE_{\sf arr}}}
\def\Earrff{{\cbE_{\sf arr, ff}}}

\def\Aarr{{\bA_{\sf arr}}}

\def\Ex{{\cE_{\sf x}}}
\def\Ey{{\cE_{\sf y}}}
\def\Ez{{\cE_{\sf z}}}

\def\bptilde{{\widetilde{\bp}}}

\def\bEcart{{\cE_{\sf cart}}}

\def\arad{{\alpha_{\sf rad}}}
\def\aang{{\alpha_{\sf ang}}}
\def\aangff{{\alpha_{\sf ang, ff}}}

\def\e0{{\epsilon_0}}

\def\iso{{\sf iso}}

\newtheorem{thm}{Theorem}
\newtheorem{prop}[thm]{Lemma}

\definecolor{redwood}{rgb}{0.67, 0.31, 0.32}

\begin{document}

\title{Electromagnetic manifold characterization of antenna arrays}

\author{Miguel R. Castellanos, \emph{Member, IEEE} and Robert W. Heath, Jr., \emph{Fellow, IEEE}
  
  \thanks{This material is based upon work supported by the National Science Foundation under grants NSF-ECCS-2153698, NSF-CCF-2225555, NSF-CNS-2147955, in part by funds from federal agency and industry partners as specified in the Resilient $\&$ Intelligent NextG Systems (RINGS) program, and in part by Samsung and the affiliates of the NSF Broadband Wireless Access Center (BWAC) I/UCRC Center award NSF-CNS-1916766.

    M. R. Castellanos, R. W. Heath are with the Department of Electrical and Computer Engineering, North Carolina State University, Raleigh, NC 27606 USA (e-mail: mrcastel@ncsu.edu; rwheathjr@ncsu.edu).}}

\maketitle 

\begin{abstract}
  Antenna behaviors such as mutual coupling, near-field propagation, and polarization cannot be neglected in signal and channel models for wireless communication. We present an electromagnetic-based array manifold that accounts for several complicated behaviors and can model arbitrary antenna configurations. We quantize antennas into a large number of Hertzian dipoles to develop a model for the radiated array field. The resulting abstraction provides a means to predict the electric field for general non-homogeneous array geometries through a linear model that depends on the point source location, the position of each Hertzian dipole, and a set of coefficients obtained from electromagnetic simulation. We then leverage this model to formulate a beamforming gain optimization that can be adapted to account for polarization of the receive field as well as constraints on the radiated power density. Numerical results demonstrate that the proposed method achieves accuracy that is close to that of electromagnetic simulations. By leveraging the developed array manifold for beamforming, systems can achieve higher beamforming gains compared to beamforming with less accurate models.
\end{abstract}


\section{Introduction}
\label{sec:introduction}

\subsection{Motivation}
Antenna arrays have evolved in past decades to enable advanced methods for wireless communication. Present day devices are typically equipped with broadband multi-antenna transceivers that enable MIMO techniques, like spatial multiplexing and precoding \cite{BjoernsonEtAlMassiveMimoIsReality2019}, and wideband communication for increased peak data rates \cite{YuanEtAlUltraWidebandMimoAntenna2020, ZhangEtAlUltraWideband8Port2019}. New antenna designs, such as dynamic metasurface antennas (DMAs) \cite{ShlezingerEtAlDynamicMetasurfaceAntennas6g2021} and liquid metal antennas \cite{BhaMaDic:RESHAPE:-A-Liquid-Metal-Based:21}, can be reconfigured to adapt to the wireless environment. The progress of wireless communication is intricately intertwined with antennas and their fundamental effect on wireless propagation.

Antenna array characteristics can be incorporated into the wireless channel to facilitate system performance analyzes and simulations. Early work on MIMO with accounted for coupling and gain patterns gains through analytical channel models with correlated fading \cite{ClerckxEtAlImpactAntennaCoupling2007,OzdemirEtAlDynamicsSpatialCorrelationAnd2004,IvrlacEtAlFadingCorrelationsWirelessMimo2003,WallaceJensenModelingIndoorMimoWireless2002}. These stochastic channel models are effective in modeling ergodic behavior and have the advantage of analytical tractability and simplicity.
Examples include the Rayleigh-fading model, the Kronecker model, and the Unitary-Independent-Unitary (UIU) channel models, all of which can be adapted or simplified as needed for analysis \cite{HeathLozanoFoundationsMIMOCommunications2018}. Parametric and deterministic models use physical modeling, site-specific measurements obtained, and ray tracing \cite{LeeEtAl28GhzMillimeterWave2018, NgEtAlEfficientMultielementRayTracing2007} to describe wireless propagation. While generally more accurate, parametric approaches are only as effective as the underlying model.  As an example, the double-directional cluster-based model in \cite{BhaOesHea:A-New-Double-Directional-Channel-Model:10} assumes that mutual coupling can be treated equally for each cluster. In reality, mutual coupling has a directional dependence due to how antenna gain patterns change when embedded into an array \cite{FriedlanderMutualCouplingMatrixArray2020}. 
In both analytical and parametric models, however, it becomes difficult to separate the antenna behavior from the wireless channel, which complicates the study of antenna effects on MIMO systems.

In this paper, we leverage electromagnetic-based antenna models incorporating fundamental physical principles to characterize the behavior of arbitrary arrays. Prior work in this context has used these types of models to analyze communication systems with large surfaces. The achievable link gain and spatial degrees-of-freedom of large intelligent surfaces were analyzed in \cite{DardariCommunicatingWithLargeIntelligent2020} through the eigenmodes of current sheets. The near-field model for communication between a large planar array and a single antenna in \cite{ BjornsonSanguinettiPowerScalingLawsAnd2020} accounted for near-field propagation and polarization mismatches. A Fourier plane-wave expansion was used in \cite{PizzoEtAlFourierPlaneWaveSeries2022, PizzoEtAlSpatiallyStationaryModelHolographic2020} to develop stochastic electromagnetic models for communication with large holographic arrays. While these studies have established fundamental behaviors of near-field communication with large surfaces, our paper addresses models for analyzing general antenna arrays. In the latter context, prior work has leveraged the concept of an extended manifold to study antennas \cite{FriedlanderExtendedManifoldAntennaArrays2020}. The extended manifold models antenna arrays as a large collection of Hertzian dipoles and has been leveraged for high-resolution direction finding \cite{FriedlanderAntennaArrayManifoldsHigh2018}, polarization sensitivity analysis \cite{FriedlanderPolarizationSensitivityAntennaArrays2019}, and near-field localization \cite{FriedlanderLocalizationSignalsNearField2019}. In this paper, we develop a manifold model that simultaneously captures near-field behavior, polarization, and mutual coupling. The model achieves better performance than prior work in approximating field strength and also incorporates how the transmit signal affects the array pattern. We also demonstrate show its application for beamforming with arbitrary antennas arrays.

\subsection{Contributions}

In this paper, we model realistic antenna arrays from an electromagnetic perspective to obtain an array manifold that maps the array beamforming weights to the electric or magnetic field at any point in space. We model each antenna by discretizing it into a large number constant current segments, each of which is treated as a Hertzian dipole. We leverage the closed-form expressions of a Hertzian dipole field to calculate the field radiated by each antenna as a linear combination of fields. To extend this to a multi-antenna array, we apply a similar principle but account for mutual coupling by determining the array current distribution when a single antenna is excited. The proposed model differs from that in \cite{FriedlanderExtendedManifoldAntennaArrays2020, FriedlanderMutualCouplingMatrixArray2020,DardariCommunicatingWithLargeIntelligent2020,FriedlanderPolarizationSensitivityAntennaArrays2019} by accounting for the near-field effects of the Hertzian dipoles, radial field components, and polarization mismatches between the different dipoles. The proposed model yields an \emph{electromagnetic array manifold} that can be leveraged for field and power density (PD) modeling. Accurate array modeling is important not only in the far-field for calculating array gains but also in the near-field for properly constraining the radiated power. Our validation results demonstrate that the proposed model effectively predicts the electric field for heterogeneous arrays. We also show an improvement over the isotropic manifold, the embedded manifold, and the extended manifold model in  \cite{FriedlanderExtendedManifoldAntennaArrays2020}.

We leverage the electromagnetic array manifold to formulate a beam pattern synthesis optimization. The proposed model allows the system to form patterns that focus energy in the near-field or steer beams in a given direction. Since the model gives the 3D components of the electric field, the optimization can also be adjusted to account for the receive polarization. An exposure constraint over an arbitrary region can also be defined in terms of radiated PD using the electromagnetic manifold. The solutions of the beamforming problem elucidate the meaning of the singular value decomposition (SVD) of the array manifold matrix. We show the left singular vectors correspond to the intrinsic polarization of the array and the right singular are associated with the uncoupled transmission modes of the antennas. The proposed manifold results in a better approximation of the measured beamforming gain. As validated in our numerical results, optimizing the transmit signal over the electromagnetic manifold therefore provides significant gains compared to the isotropic model.

\subsection{Connections to prior work}

A large portion of prior work on array beamforming has focused on isotropic radiators or on simple models for how the antenna affects the beam shape. MIMO systems leverage beamforming and precoding to maximize the system spectral efficiency. In single-stream communication, optimal beamforming entails maximizing the received power. A number of array processing and wireless studies analyze beamforming from the perspective of isotropic radiators \cite{LebretBoydAntennaArrayPatternSynthesis1997,CaoEtAlConstantModulusShapedBeam2017,AlkhateebEtAlChannelEstimationAndHybrid2014,LoveEtAlGrassmannianBeamformingMultipleInput2003,RaghavanEtAlSystematicCodebookDesignsQuantized2007}. Isotropic radiators are a reasonable assumption when the antenna pattern is nearly omnidirectional or when the antenna gain pattern is included in the channel model. In the latter case, the beam patterns are design with the antenna effects in mind, even if not explicitly considered. Modern approaches to beamforming also include some of the antenna characteristics, including antenna element patterns \cite{GemechuEtAlBeampatternSynthesisWithSidelobe2020,CuiEtAlEffectiveArtificialNeuralNetwork2021}, mutual coupling \cite{ZhangSerRobustBeampatternSynthesisAntenna2011, SchmidEtAlEffectsCalibrationErrorsAnd2013}, antenna polarization \cite{XiaoNehoraiOptimalPolarizedBeampatternSynthesis2009, FuchsFuchsOptimalPolarizationSynthesisArbitrary2011}, and near-field effects \cite{MyersEtAlNearFieldFocusingUsing2022}. By modeling all of the effects at once, MIMO systems can leverage the full capabilities of advanced antenna structures \cite{CastellanosHeathLinearPolarizationOptimizationWideband2023, ShyianovEtAlAchievableRateWithAntenna2022,JensenWallaceCapacityContinuousSpaceElectromagnetic2008}. The proposed model accurately incorporates antenna patterns, mutual coupling, and polarization in a way that is tractable for near-field and far-field beamforming with arbitrary antennas. 

User electromagnetic exposure levels are affected by the near-field behavior of antennas. All wireless devices are heavily regulated to ensure that consumers are not subject to radiation levels above safety thresholds. Prior work has shown that user exposure constraints can deteriorate the beamforming gain and achievable rate of wireless systems \cite{ying2015closed,hochwald2013sar}. Fortunately, a variety of studies have shown that exposure models can be leveraged to proactively design transmit signals to increase throughput while satisfying exposure constraints \cite{castellanos2021dynamic,ying2017sum,castellanos2016hybrid,ying2013beamformer,HelJamBro:Exposure-Modelling-and-Minimization:20}. In this context, exposure modeling is paramount to incorporating signal design optimization problems with the desired exposure limit. Similarly to channel models, exposure models generally rely on extensive simulation-based fitting \cite{ebadi2018determining,li2017high,hochwald2014incorporating} or on approximations of the near-field radiation \cite{castellanos2019signal}. Exposure-aware signal processing requires an effective model that remains analytically tractable. The proposed manifold model can be simultaneously leveraged for PD modeling and for solving the exposure-aware beamforming problem.

\subsection{Organization and notation}

The remainder of this paper has the following organization. In Section \ref{sec:system-model}-(a) we develop the antenna by discretizing a radiating volume as a collection of Hertzian dipoles. We then apply linearity and superposition to derive an array model in Section \ref{sec:system-model}-(b). In Section \ref{sec:system-model}-(c), we analyze the behavior of the proposed model in the far-field. In Section \ref{sec:validation}, we compare the proposed model with electromagnetic simulation results. The beamforming optimization problem and solutions are discussed in Section \ref{sec:beam-pattern}. Numerical beam pattern results are shown in Section \ref{sec:simulations}. Key insights and extensions for future work are addressed in Section \ref{sec:conclusion}.

A column vector is denoted as bold lowercase letter $\ba$, and a matrix is denoted as a bold uppercase letter $\bA$. For distinction, all field scalars are denoted as script uppercase letters $\cA$ and field vectors are denoted as bold script uppercase letters $\cbA$. The transpose of $\bA$ is denoted as $\bA^{\T}$, the conjugate of $\bA$ is denoted as $\bA^{\sf c}$, and the conjugate transpose of $\bA$ is denoted as $\bA^*$. An $N \times N$ diagonal matrix with entries $d_1, d_2, \cdots, d_N$ is denoted as $\diag(d_1, d_2, \cdots, d_N)$ and a block diagonal matrix with matrices $\bA_1, \bA_2, \cdots, \bA_N$ along the diagonal is denoted as $\blkdiag(\bA_1, \bA_2, \cdots, \bA_N)$. The two-norm of a vector is denoted as $\norm{\ba}$. The Frobenius norm of $\bA$ is denoted as $\norm{\bA}_{\sf F}$. The operator $\nabla \times$ denotes the curl operator of a vector field. The unit imaginary number is denoted as $\cj$. For a complex number $z = x + \cj y$, the real part of $z$ is denoted as $\Re(z) = x$ and the imaginary part of $z$ as $\Im(z) = y$.

\section{Proposed model development and analysis}
\label{sec:system-model}

In this section, we describe the antenna and the array models. We briefly discuss how the electromagnetic response of the antenna is calculated using the surface currents and how this relates to the Hertzian dipole interpretation. Then we present the electromagnetic array manifold.

\subsection{Antenna model}

We begin by describing the radiation from a single antenna element centered at the origin operating at frequency $f$. In the following, we will denote $c$ as the speed of light, $\lambda = c/f$ as the wavelength, $\omega = 2 \pi f$ as the radian frequency, and $\beta = 2 \pi / \lambda$ as the angular wavenumber. We additionally define $\epsilon_0$ as the permittivity of free-space, and $\mu_0$ as the permeability of free-space. All time-varying variables, such as fields and currents, are represented in phasor notation. This model is discussed in the context of narrowband operation but could be extended to a wideband case by including the frequency dependence of all relevant quantities.

Antennas radiate electromagnetic fields due to the time-varying currents in a conductive medium. Let $V \subseteq \R^3$ denote the spatial region encompassing the antenna and let $\cJ(\bp)$ denote the current density at point $\bp \in V$. The radiated fields are typically computed using an auxiliary quantity known as the magnetic vector potential $\cbA(\bp)$, given by \cite{StutzmanThieleAntennaTheoryAndDesign2012}
\begin{equation}
  \label{eq:vector-potential}
  \cbA(\bp) = \mu_0 \int_\cV \frac{e^{\cj \beta \, \abs{\bp - \bp'}}}{4 \pi \abs{\bp - \bp'}} \cJ(\bp') d \bp'.
\end{equation}
The magnetic vector potential can be used to calculate the electric  $\cbE(\bp)$ and the magnetic field $\cH(\bp)$ as
\begin{equation}
  \label{eq:A2H}
  \cH(\bp) = \frac{1}{\mu_0}\nabla \times \cbA(\bp),
\end{equation}
\begin{equation}
  \label{eq:A2E}
  \cbE(\bp) = \frac{1}{\cj \omega \epsilon_0} \nabla \times \cH(\bp).
\end{equation}
Obtaining the fields is straightforward given $\cbA(\bp)$, but the integral in \eqref{eq:vector-potential} is generally difficult to evaluate. 

Rather than assuming a particular current distribution and attempting to obtain $\cbA(\bp)$ from \eqref{eq:vector-potential}, we will use a discretization approach to get a more tractable expression for $\cbA(\bp)$. Partition $V$ into $K$ regions with the $k$th region denoted as $V_k \in V$. The surface current over the $k$th region can be assumed to be constant if the regions are sufficiently small. Mathematically, we express this as $\cJ(\bp) = \cJ_k$ for $\bp \in V_k$. The vector potential can then be written as
\begin{align} 
  \label{eq:vector-potential-partition}
  \cbA(\bp) \approx \mu_0 \sum_{k=1}^{K} \int_{V_k} \frac{e^{\cj \beta \, \abs{\bp - \bp'}}}{4 \pi \abs{\bp - \bp'}} \cJ_k d \bp'.
\end{align}
We further assume that the volumes $V_k$ are small enough and that the point $\br$ is sufficiently far away from each $V_k$ such that $\bp - \bp' = \bp_k$ for $\bp' \in V_k$. Letting $\bs_k$ denote the centroid of $V_k$, we will further assume that $\bp_k = \bp - \bs_k$. Under these assumptions, the integral in \eqref{eq:vector-potential-partition} can be simplified as 
\begin{equation}
  \label{eq:20}
  \int_{V_k} \frac{e^{\cj \beta \, \abs{\bp - \bp'}}}{4 \pi \abs{\bp - \bp'}} \cJ_k d \bp' =  \frac{e^{\cj \beta \, \abs{ \bp_k}}} {4 \pi \abs{\bp_k}} \cJ_k \int_{V_k}  d \bp' = \frac{e^{\cj \beta \, \abs{ \bp_k}}} {4 \pi \abs{\bp_k}} \cJ_k \abs{V_k}
\end{equation}
Defining the effective dipole moment of the $k$th antenna segment as $\bmm_k = \cJ_k \abs{V_k}$, this gives
\begin{align}
  \label{eq:disc-vector-potential}
  \cbA(\bp) & \approx \sum_{k=1}^{K} \frac{\mu_0 \abs{V_k} \cJ_k}{4 \pi \abs{\bp_k}} e^{\cj \beta \, \abs{\bp_k}} \\
  & = \sum_{k=1}^{K} \underbrace{\frac{\mu_0 \bmm_k}{4 \pi \abs{\bp_k}} e^{\cj \beta \, \abs{\bp_k}}}_{\cbA_k(\bp)}. \label{eq:disc-vector-potential}
\end{align}
We denote the term within the summation of \eqref{eq:disc-vector-potential} as the $k$th vector potential $\cbA_k(\bp)$. The principle of superposition allows us to find the radiated fields as the sum of the contributions from the $K$ vector potentials.

We further simplify the expression for the radiated fields by noting that $\cbA_k(\bp)$ has the same form as the vector potential of a Hertzian dipole. We first define a few variables that will aid in the field equations. We express the point $\bp$ using spherical coordinates with radial distance $r$, azimuthal angle $\phi$ and elevation angle $\theta$ as
\begin{equation}
  \label{eq:point-sph}
  \bp = r[\cos \phi \sin\theta, \, \sin \phi \sin \theta, \, \cos \theta]^{\T}.
\end{equation}
We follow the derivation in \cite{StutzmanThieleAntennaTheoryAndDesign2012} to express the Hertzian dipole radiated fields in terms of its moment $\bmm$. The  unit vectors corresponding to the spherical orthonormal basis are
\begin{align}
  \label{eq:unit-sph-vecs}
  \bur & = \left[ \cos \phi \sin\theta, \, \sin \phi \sin \theta, \, \cos \theta \right]^{\T}, \nonumber \\
  \bue & = \left[ -\sin \phi, \, \cos \phi, \, 0 \right]^{\T}, \\
    \bua & = \left[ -\cos \phi \sin \theta, \, -\sin \phi \sin \theta, \, \cos \theta \right]^{\T}. \nonumber
  \end{align}
  The radiated electric field of a Hertzian dipole is a function of the dipole moment and the point $\bp$. The amplitudes of the field components decay with distance according to the radial and angular amplitudes $\arad(\bp)$ and $\aang(\bp)$:
  \begin{equation}
    \label{eq:rad-amp}
    \arad(\bp) = \frac{e^{-\cj \beta r}}{\cj \omega \e0 2 \pi}\left(\frac{1}{r^3} + \frac{\cj \beta}{r^2} \right),
  \end{equation}
  \begin{equation}
    \label{eq:ang-amp}
    \aang(\bp) = -\frac{e^{-\cj \beta r}}{\cj \omega \e0 4 \pi} \left(\frac{1}{r^3} + \frac{\cj \beta}{r^2} - \frac{\beta^2}{r} \right).
  \end{equation}
  Then the spherical basis components of the electric field induced by a current dipole moment $\bm$ can be written as \cite{StutzmanThieleAntennaTheoryAndDesign2012}
  \begin{align}
    \Edipr(\bmm, \bp) & = \arad \bmm^{\T} \bur(\bp), \nonumber \\ 
    \Edipe(\bmm, \bp) & = \aang \bmm^{\T} \bue(\bp), \label{eq:E-dip-comp} \\
    \Edipa(\bmm, \bp) & = \aang \bmm^{\T} \bua(\bp).  \nonumber
  \end{align}
  The dipole field $\bEdip$ in spherical coordinates is then given by
  \begin{equation}
    \label{eq:21} 
    \bEdip(\bmm, \bp) = [\Edipr(\bmm, \bp), \, \Edipe(\bmm, \bp), \, \Edipa(\bmm, \bp)]^{\T}.
  \end{equation}
  The dipole field contains contributions in both angular directions and in the radial direction depending on the moment $\bm$ and the distance from the dipole. This model accounts for polarization because the direction of $\bmm$ at each segment will affect the orientation of $\bEdip$. Fig. \ref{fig:antenna-model} shows how an example of the segmenting and field calculation process for a linear antenna. The antenna is divided into a $K$ segments, and each segment contributes to the electric field at point $\bp$.
  \begin{figure}
  \centering
  \includegraphics[width = 0.4\textwidth]{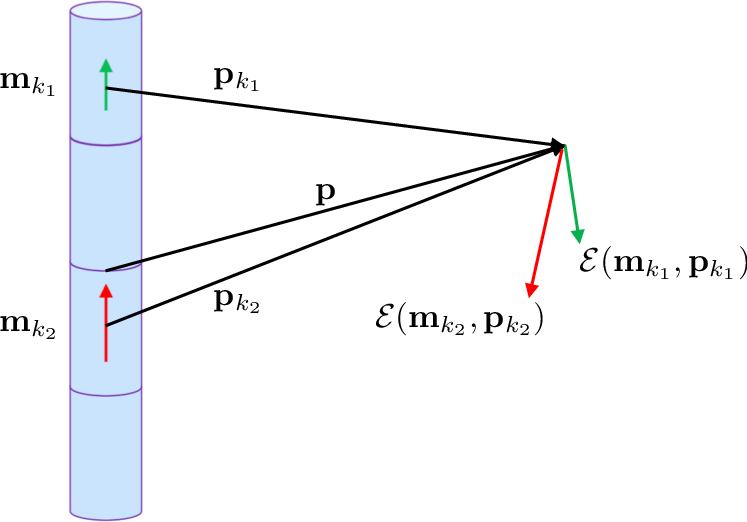}
  \caption{\label{fig:antenna-model} Diagram of the partitioned antenna model used to calculate the electric field at a point $\bp$. The moment vector $\bmm_k$ and the relative position $\bp_k$ of each segment is used to calculate the electric field contribution $\bE(\bmm_k, \bp_k)$.}
\end{figure}

  The dipole field expressions in \eqref{eq:E-dip-comp} are given in local spherical coordinates and must be rotated to apply superposition. We let the Cartesian basis components of the field be defined as $\Ex(\bmm, \bp)$, $\Ey(\bmm, \bp)$, and $\Ez(\bmm, \bp)$ and the corresponding field vector as $\bEcart(\bmm, \bp) = [\Ex(\bmm, \bp), \, \Ey(\bmm, \bp), \, \Ez(\bmm, \bp)]^{\T}$. The mapping between the spherical and Cartesian bases can be represented by transformation matrix $\bQ(\bp) = [\bur(\bp), \, \bua(\bp), \, \bue(\bp)]$, which yields 
  \begin{equation}
    \label{eq:sph2cart}
    \bEcart(\bmm, \bp) = \bQ(\bp) \cbE(\bmm, \bp)
  \end{equation}
  The rotation matrix is dependent on $\bp$ because the spherical basis is as function of the angles $\theta$ and $\phi$. 

  We now use the previous results to find the radiated electric field. Combining the results from \eqref{eq:disc-vector-potential}, \eqref{eq:E-dip-comp}, and \eqref{eq:sph2cart} the total field (in a spherical basis) is
  \begin{equation}
    \label{eq:E-sum-ant}
    \cbE(\bp) = \bQ^{\T}(\bp) \sum_{k=1}^{K} \bQ(\bp_k) \bEdip(\bmm_k, \bp_k).
  \end{equation}
  To simplify \eqref{eq:E-sum-ant}, we define the following matrices: the dipole field transform
  \begin{equation}
    \label{eq:dipole-field-trans}
    \bT_k(\bp) = \left[
      \begin{array}{ccc}
        \arad(\bp_k) \bur(\bp_k) & \aang(\bp_k) \bua(\bp_k) & \aang \bue(\bp_k)
      \end{array}
    \right]^{\T}
  \end{equation}
  and the rotational coherence matrix $\bR^{\T}_k(\bp) = \bQ^{\T}(\bp) \bQ(\bp_k)$. These definitions yield
  \begin{equation}
    \label{eq:E-dip-lin}
    \bEdip(\bmm_k, \bp) = \bR^{\T}_k(\bp) \bT_k(\bp) \bmm_k.
  \end{equation}
  We also define the antenna rotational coherence matrix $\Rbar^{\T}(\bp) = [\bR^{\T}(\bp,\bp_1), \, \cdots, \, \bR^{\T}(\bp,\bp_{K})]$, the antenna dipole field transform $\Tbar(\bp) = \blkdiag(\bT(\bp_1), \, \cdots, \, \bT(\bp_{K}))$, and the antenna moment vector $\mbar = [\bmm_1^{\T}, \, \cdots, \, \bmm_{K}^{\T}]^{\T}$. These definitions yield
  \begin{equation}
    \label{eq:E-lin-ant}
    \cbE(\bp) = \Rbar^{\T}(\bp) \Tbar(\bp) \mbar.
  \end{equation}
  The principle of superposition and the simple expressions for the field of each antenna segment yield a linear relationship between the moments and the radiated field.

  The feed current and the type of feed, which have so far been neglected, affects the antenna current distribution and the radiated field. The feed current can be assumed to scale the current moments. Let $w$ be the complex feed current weight and assume that $\mbar$ represents the effective current vector for a unity current. The electric field radiated by the antenna excited by $w$ is
  \begin{equation}
    \label{eq:E-lin-ant-weighted}
    \cbE_{\sf ant}(\bp, w) = \Rbar^{\T} \Tbar \mbar w.
  \end{equation}
  The moment vector is determined by the antenna design, the type of feed, and objects in near proximity to the antenna. By defining the moment vector with respect to a unity current, however, $\mbar$ is independent of both $\bp$ and $w$. The antenna moment can be obtained from electromagnetic simulation software. The choice of $K$ is typically determined by the software through a meshing procedure but can also be changed manually to reduce the complexity of the model.

  \subsection{Array model}
  \label{sec:array-model}

In the previous section, we analyzed the radiated field of a single antenna excited by a feed current. In this section, we apply the same techniques to find the radiated field from an array that is used for beamforming.

We apply the same arguments used before to model the array electric field by partitioning the array. The antenna array is modeled as a disjoint volume $V$. As before, we partition the $n$th antenna into $K_n$ pieces, denoted as $V_{n,k}$. Let $\bw \in \C^{N}$ be the beamforming vector which represents the weights applied to the excitation current of each array element. The excitation currents are assumed to be unity for simplicity. The field radiated by the array will depend on the current distribution in all of the antennas, which is dependent on the array excitation. We let $\bmm_{n,k}(\bw)$ denote the dipole moment induced in the $k$th segment of the $n$th antenna when the array is excited by $\bw$.  We similarly define $\bp_{n, k}$ as the position vector from the centroid of the $k$th segment of the $n$th antenna to $\bp$. We define the antenna moment vector for the $n$th antenna as
  \begin{equation}
    \mbar_n(\bw) = [\bmm_{n,1}^{\T}(\bw), \, \cdots, \, \bmm_{n,K_n}^{\T}(\bw)]^{\T}.
  \end{equation}
  The antenna moment vectors are then stacked together to form the array moment vector 
  \begin{equation}
    \mbar_{\sf arr}(\bw) = [\mbar_{n}^{\T}(\bw), \, \cdots, \, \mbar_{0,N-1}^{\T}(\bw)]^{\T}.
  \end{equation}
   Letting the total number of segments be $\Kbar = \sum_{n=1}^{N}K_n$, $\mbar_{\sf arr}(\bw)$ is a $3 \Kbar$ length vector because each $\bmm_{n,k}(\bw)$ has three spatial components. The only physical difference between the array and antenna models is that the volume is not contiguous, but this does not functionally affect the result.

  The array model approximates the array as an extended array of Hertzian dipoles with current $\bmm_{n,k}(\bw)$ centered at $\bp_{n, k}$. To convert the dipole currents to fields, we define the $n$th dipole field transform $\Tbar_n(\bp) = \blkdiag(\bT(\bp_{n,1}), \cdots, \bT(\bp_{n,K_n})) \in \C^{3K_n \times 3K_n}$ and the $n$th antenna rotational coherence matrix $\Rbar^{\T}_n(\bp) = [\bR^{\T}(\bp, \bp_{n,1}), \, \cdots, \, \bR^{\T}(\bp, \bp_{n,K_n})]^{\T} \in \R^{3 \times 3K_n}$. We further define block matrices  $\Rarr^{\T}(\bp) = [\Rbar^{\T}_1(\bp), \, \cdots, \, \Rbar^{\T}_{N}(\bp)]$, denoted as the $3 \times \Kbar$ array rotational coherence matrix, and $\Tarr(\bp) = \blkdiag(\Tbar_1(\bp), \, \cdots, \, \Tbar_{N}(\bp))$, denoted as the $3\Kbar \times 3\Kbar$ array dipole field transform. The array rotational coherence matrix and the array dipole field transform do not depend on the feed currents, so they remain constant regardless of which antennas are excited. The effective moment vector, however, encapsulates the current distribution throughout the entire array due to $\bw$. The array electric field is obtained through the summation of the contributions from the extended dipole array as
  \begin{equation}
  \label{eq:E-array}
  \Earr(\bp, \bw) = \Rarr^{\T}(\bp) \Tarr(\bp) \mbar_{\sf arr}(\bw).
\end{equation}
One issue with this model is that the array currents would need to be recalculated depending on the excitation. In the following, we separate the effect of $\bw$ from the currents.

Linearity can be applied to isolate $\bw$ from the moment vector by referencing the moments to unity currents. The main idea of the proposed array model is shown in Fig. \ref{fig:array-model-diagram}. Let $\bee_n \in \C^{N}$ be the standard unit vector with $n$th entry equal to one. We denote the current in the $n$th antenna when the $\ell$th antenna is excited by a unity current as $\mbar^{(\ell)}_n = \mbar_n(\bee_\ell)$. As shown n Fig. \ref{fig:array-model-diagram}-(a), the excitation of the $\ell$th antenna by $w_\ell$ leads to a current distribution that is captured by the moment vector $w_{\ell} \mbar_\ell^{(\ell)}$. Due to mutual coupling, however, the $m$th element in the array ($m \neq \ell$) also exhibits a current distribution characterized by $w_{\ell} \overline{\bmm}^{(\ell)}_{m}$.
Let $\mbar_{\sf arr}^{\ell} =  [(\mbar_{1}^{(\ell)})^{\T}, \, \cdots, \, (\mbar_{N}^{\ell})^{\T}]^{\T}$. The resulting electric field when the $\ell$th antenna is excited by $w_\ell$ is
\begin{equation}
  \label{eq:lin-array-indiv}
  \cbE_{\sf arr}^{(\ell)}(\bp, w_\ell) = \Rarr^{\T}(\bp) \Tarr(\bp)  \mbar_{\sf arr}^{(\ell)} w_\ell
\end{equation}
where  $\cbE^{(\ell)}(\bp, w_{\ell})$ is obtained by weighing each Hertzian dipole field contribution by $w_\ell$. To calculate the total field, we simply need to sum the field contributions when each antenna is individually excited. 
To find the electric field when the entire array is excited, we define $\Mbar = [\mbar_{\sf arr}^{(0)}, \, \cdots, \, \mbar_{\sf arr}^{(N-1)}]$ as the $3\Kbar \times N$ effective array moment matrix. From linearity, we have $\mbar_{\sf arr}(\bw) = \Mbar \bw$, which gives
\begin{equation}
  \label{eq:E-lin-array}
  \Earr(\bp, \bw) = \Rarr^{\T}(\bp) \Tarr(\bp) \Mbar \bw.
\end{equation}
The total current in the array is the sum of the currents when each element is excited individually as shown in Fig. \ref{fig:array-model-diagram}. 

\begin{figure}
  \centering
  \begin{subfigure}{0.4\textwidth}
    \centering
    \includegraphics[width=\textwidth]{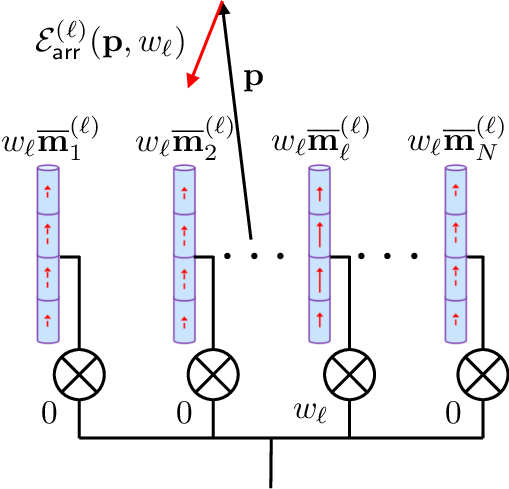}
    \caption{}
  \end{subfigure}
  \hfill
  \begin{subfigure}{0.4\textwidth}
    \centering
    \includegraphics[width=\textwidth]{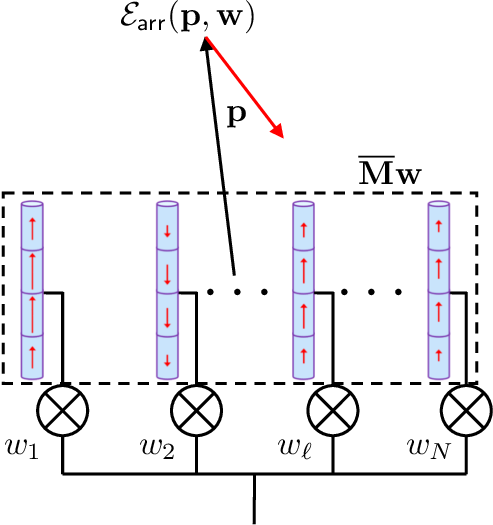}
    \caption{}
  \end{subfigure}
  \caption{\label{fig:array-model-diagram} Diagram of the discretized array model. (a) The $\ell$th antenna in the array is excited by a unit current weighed by $w_{\ell}$, which induces an antenna moment vector of $w_\ell \mbar_{n}^{\ell}$ in the $n$th element due to coupling. The fields radiated by each antenna result in $\cbE^{(\ell)}(\bp, w_{\ell})$. (b) All of the elements in the array are excited. The matrix $\overline{\bM}$ captures the moments of each segment of each antenna element to calculate the field $\cbE(\bp, \bw)$.}
\end{figure} 

The array can also be characterized in terms of an array manifold that maps the weights to the radiated fields. We let the $N \times 3$ array manifold matrix be defined as
\begin{equation}
  \Aarr^{\T}(\bp) =  \Rarr^{\T}  \Tarr \Mbar.
\end{equation}
Defining the manifold in this way yields a simple relationship for calculating the electric field, given as
\begin{equation}
  \label{eq:manifold-field}
  \Earr(\bp) = \Aarr^{\T}(\bp) \bw.
\end{equation}
Each column in $\Aarr^{\T}(\bp)$ maps to a spatial component of the electric field. Let $\ba_{\sf az}(\bp)$ be the steering vector for the azimuth direction, and let $\ba_{\sf el}(\bp)$ and $\ba_{\sf rad}(\bp)$ be defined for the elevation and radial directions. The matrix $\bA(\bp)$ can also be written as $\bA(\bp) = [\ba_{\sf az}(\bp), \, \ba_{\sf el}(\bp), \, \ba_{\sf rad}(\bp)]^{\T}$. Each of these vectors corresponds to a different polarization component of the array; $\ba_{\sf az}(\bp)$ corresponds to horizontal polarization, $\ba_{\sf el}(\bp)$ corresponds to vertical polarization, and $\ba_{\sf rad}(\bp)$ corresponds to the radial polarization. The steering vectors can be used to model individual polarizations. The application of the manifold for beamforming is shown in Section IV.

In most scenarios, exposure regulations constrain either the specific absorption rate (SAR) or the plane-wave equivalent PD \cite{lin2020fcc}. Letting $\eta_0$ denote the impedance of free-space, the plane-wave equivalent PD is
\begin{equation} 
  \label{eq:1}
  {\sf PD}(\bp) = \frac{\norm{\Earr(\bp)}^2}{2 \eta_0}.
\end{equation} 
The proposed model can be used to find the plane-wave equivalent PD. Using the relationship in \eqref{eq:manifold-field} for $\cbE_{\sf arr}(\bp)$
\begin{equation}
  \label{eq:pw-pd-model}
  {\sf PD}(\bp) = \frac{1}{2 \eta_0} \bw^* \Aarr^*(\bp) \Aarr(\bp) \bw,
\end{equation}
which is a quadratic form in terms of the beamformer $\bw$. This PD expression yields a mathematically tractable model that can be incorporated into the signal design to constrain the system exposure. A technique similar to that in \cite{castellanos2019signal} could be applied to transform the PD model into a SAR model, but we leave that analysis for future work. 

\subsection{Far-field approximation}

The prior analysis does not make any assumptions regarding the location of $\bp$ relative to the array. We now analyze the behavior of the proposed model in the far-field.  Consider a Hertzian dipole located at $\bs$ with moment vector $\bmm$ and relative position $\bptilde = \bp - \bs$. From \eqref{eq:E-dip-comp} and\eqref{eq:sph2cart}, the field at point $\bp$ in global coordinates is given by
\begin{equation}
  \label{eq:nf-dipole}
  \cbE_{\bs, \sf nf}(\bmm, \bp) = \bR^{\T}(\bp, \bptilde) \bT(\bptilde) \bmm.
\end{equation}
The following Lemma characterizes the accuracy of  model $\bp$ moves away from the antennas. 

\begin{prop}
  \label{prop:ff-dipole}
  Consider a Hertzian dipole located at $\bs$ with moment vector $\bmm$ and let $\bp$ be a position vector with norm $r$ as in \eqref{eq:point-sph}. In addition, let $\bptilde = \bp - \bs$. Define
    \begin{equation}
    \label{eq:ff-ang-amp}
    \aangff(\bp) = \frac{\beta^2 e^{-\cj \beta r}}{\cj \omega \e0 4 \pi r},
  \end{equation}
  and let the far-field dipole transform be defined as
  \begin{equation}
    \label{eq:ff-dipole-transform}
    \Tff(\bp) = \left[
      \begin{array}{ccc}
        {\bf 0} & \aangff(\bp) \bua(\bp) & \aangff(\bp) \bue(\bp)
      \end{array}
    \right]^{\T}.
  \end{equation}
  Then the relative error between the radiated field $\cbE_{\bs, \sf nf}(\bmm, \bp)$ and the far-field approximation
  \begin{equation}
    \label{eq:ff-dipole}
    \cbE_{\bs, \sf ff}(\bmm, \bp) = \Tff(\bptilde) \bmm
  \end{equation}
  converges to zero as 
  \begin{equation}
    \label{eq:rel-error-nf-ff-dipole}
    \lim\limits_{r \to \infty} \frac{\norm{\cbE_{\bs, \sf nf}(\bmm, \bp) - \cbE_{\bs, \sf ff}(\bmm, \bp)}}{\norm{\cbE_{\bs, \sf nf}(\bmm, \bp)}} = 0
  \end{equation}
\end{prop}

\begin{proof}
  We first substitute \eqref{eq:nf-dipole} and \eqref{eq:ff-dipole} into the relative error in \eqref{eq:rel-error-nf-ff-dipole} to obtain
  \begin{align}
     \frac{\norm{\cbE_{\bs, \sf nf}(\bmm, \bp) - \cbE_{\bs, \sf ff}(\bmm, \bp)}}{\norm{\cbE_{\bs, \sf nf}(\bmm, \bp)}} = \frac{\norm{\left(\bR^{\T}(\bp, \bptilde) \bT(\bptilde) - \Tff(\bptilde)\right) \bmm}}{\norm{\cbE_{\bs, \sf nf}(\bmm, \bp)}}.
  \end{align}
  We further apply the Cauchy-Schwarz inequality to get
  \begin{align}
    \label{eq:rel-error-ineq}
    \frac{\norm{\left(\bR^{\T}(\bp, \bptilde) \bT(\bptilde) - \Tff(\bptilde)\right) \bmm}}{\norm{\cbE_{\bs, \sf nf}(\bmm, \bp)}} \leq \nonumber \\ \frac{\norm{\bR^{\T}(\bp, \bptilde) \bT(\bptilde) - \Tff(\bptilde)}_{\sf F} \norm{\bmm}}{\norm{\cbE_{\bs, \sf nf}(\bmm, \bp)}}. 
  \end{align}
  To prove the error converges to zero, we look at the denominator and numerator separately. In the far-field, the electric field from a Hertzian dipole decreases as $1/r$. Because of this, $r \norm{\cbE_{\bs, \sf nf}(\bmm, \bp)}$ converges to a constant as $r$ tends to infinity. It therefore suffices to show that the term $r \norm{\bR^{\T}(\bp, \bptilde) \bT(\bptilde) - \Tff(\bptilde)}_{\sf F}$ converges to zero.
  
  We analyze the behavior of $\bptilde$ by using a first-order Taylor expansion for $\norm{\bptilde}$ as 
  \begin{align}
    \label{eq:dist-approx}
    \norm{\bp - \bs} & = \sqrt{r^2 - 2\bp^{\T}\bs + \norm{\bs}^2} \nonumber \\
                       & = r - \frac{\bp^{\sf T} \bs}{r} + O \left( \frac{1}{r^2} \right).
  \end{align}
  Let $\varphi$ be the angle between the vectors $\bp$ and $\bptilde$ which satisfies
  \begin{equation}
    \label{eq:5}
    \cos \varphi = \frac{\bptilde^{\T} \bp}{\norm{\bptilde} \norm{\bp}}.
  \end{equation}
  Substituting the approximation in \eqref{eq:dist-approx} yields
  \begin{equation}
    \label{eq:6}
    \cos \varphi = \frac{r^2 - \bp^{\T} \bptilde}{r^2 - \bp^{\T} \bptilde + O(1/r)} = 1 + O(1/r^3).
  \end{equation}
  This implies that for large $r$, the angular separation between $\bp$ and $\bptilde$ can be approximated as $\varphi \approx 0$. From the definition of $\bR(\bp, \bptilde)$, this gives $\bR(\bp, \bptilde) \rightarrow \bI$ elementwise.

  We now look at the dipole transform $\bT(\bp)$ and show that $r\bT(\bptilde) \rightarrow r\Tff(\bptilde)$. From the previous argument regarding the angular separation between $\bp$ and $\bptilde$, any spherical basis vector in \eqref{eq:unit-sph-vecs} converges elementwise as $\bu(\bptilde) \rightarrow \bu(\bp)$. From the expressions for $\arad$ and $\aang$ in \eqref{eq:rad-amp} and \eqref{eq:ang-amp}, we also have $r\arad(\bptilde) \rightarrow 0$ and $r\aang(\bptilde) \rightarrow r\aangff(\bptilde)$. This implies the elementwise limit $r\bT(\bptilde) \rightarrow r\Tff(\bptilde)$ as $r$ goes to infinity.

  The limits $\bR(\bp, \bptilde) \rightarrow \bI$ and $r\bT(\bptilde) \rightarrow r\Tff(\bp)$ imply
  \begin{equation}
    \label{eq:ineq-limit}
    \lim\limits_{r \to \infty} \frac{r\norm{\bR^{\T}(\bp, \bptilde) \bT(\bptilde) - \Tff(\bp)}_{\sf F} \norm{\bmm}}{r\norm{\cbE_{\bs, \sf nf}(\bmm, \bp)}} = 0. 
  \end{equation}
  The inequality in \eqref{eq:rel-error-ineq} and the limit \eqref{eq:ineq-limit} imply \eqref{eq:rel-error-nf-ff-dipole}, which completes the proof.
\end{proof}

The previous result can be applied to develop a far-field model. The array model for the radiated electric field can be expressed as the summation
\begin{equation}
  \label{eq:arr-nf-sum}
  \cbE_{\sf arr}(\bp, \bw) = \sum_{n=1}^{N} \sum_{k=1}^{K_n} \bR(\bp, \bp_{n,k}) \bT(\bp_{n,k}) \bmm_{n,k}(\bw).
\end{equation}
For $\bp$ sufficiently far from the array, the array field can be approximated as
\begin{equation}
  \label{eq:arr-ff-sum}
  \cbE_{\sf arr, ff}(\bp, \bw) = \sum_{n=1}^{N} \sum_{k=1}^{K_n} \Tff(\bp_{n,k}) \bmm_{n,k}(\bw)
\end{equation}
as shown in the following lemma.

\begin{prop}
  \label{prop:ff-array}
  Consider $N$ antennas excited by the weights $\bw$, and the $n$th antenna be composed of $K_n$ Hertzian dipoles. Let the $k$th dipole in the $n$th antenna be located at $\bs_{n,k}$ with moment vector $\bmm_{n,k}$ and let $\bp$ be a position vector with norm $r$ as in \eqref{eq:point-sph}. In addition, let $\bp_{n,k} = \bp - \bs_{n,k}$. Then the relative error between $\cbE_{\sf arr}(\bp, \bw)$ and $\cbE_{\sf arr, ff}(\bp, \bw)$ converges to zero as
  \begin{equation}
    \label{eq:rel-error-nf-ff-array}
    \lim\limits_{r \to \infty} \frac{\norm{\cbE_{\sf arr}(\bp, \bw) - \cbE_{\sf arr, ff}(\bp, \bw)}}{\norm{\cbE_{\sf arr}(\bp, \bw)}} = 0
  \end{equation}
  as long as $r\lim_{r \to \infty} \norm{\cbE_{\sf arr}(\bp, \bw)}$ converges to a constant.
\end{prop}
\begin{proof}
  We substitute \eqref{eq:arr-nf-sum} and \eqref{eq:arr-ff-sum} in the error to obtain
  \begin{align}
    \label{eq:arr-sum-ineq}
    & \frac{\norm{\cbE_{\sf arr}(\bp, \bw) - \cbE_{\sf arr, ff}(\bp, \bw)}}{\norm{\cbE_{\sf arr}(\bp, \bw)}} = \nonumber \\
    & \frac{\norm{\sum\limits_{n=0}^{N-1} \sum\limits_{k=0}^{K_n-1} \left(\bR(\bp, \bp_{n,k}) \bT(\bp_{n,k}) - \Tff(\bp_{n,k}) \right) \bmm_{n,k}(\bw)}}{\norm{\cbE_{\sf arr}(\bp, \bw)}} \leq \nonumber \\
      & \frac{\sum\limits_{n=0}^{N-1} \sum\limits_{k=0}^{K_n-1} \norm{\left(\bR(\bp, \bp_{n,k}) \bT(\bp_{n,k}) - \Tff(\bp_{n,k}) \right) \bmm_{n,k}(\bw)}}{\norm{\cbE_{\sf arr}(\bp, \bw)}}.
  \end{align}
  The rest of the proof follows the arguments in Lemma 1 to show that the numerator summands scaled by $r$ converge to $0$. Together with the assumption that $r\lim_{r \to \infty} \norm{\cbE_{\sf arr}(\bp, \bw)}$, this implies the limit in \eqref{eq:rel-error-nf-ff-array} holds.
\end{proof}

The far-field model can be expressed in a similar form as \eqref{eq:E-lin-array}. One key observation from Lemma 2 is that the rotational coherence matrix is not needed in the far-field model. In other words, the spherical basis vectors observed at each $\bp_{n,k}$ will appear identical as $\bp$ moves far from the array. It suffices to then define a far-field array dipole field transform as $\Tarrff = \left[ \Tff(\bp_{1,1}), \, \cdots, \, \Tff(\bp_{1,K_1}), \, \cdots, \, \Tff(\bp_{N, K_N}) \right]$ without the rotations. The far-field model is then
\begin{equation}
  \label{eq:7}
  \cbE_{\sf arr, ff}(\bmm, \bp) = \Tarrff \Mbar \bw,
\end{equation}
and the far-field array manifold matrix is given by
\begin{equation}
  \label{eq:3}
  \bA_{\sf arr, ff}^{\T} = \Tarrff \Mbar.
\end{equation}
We will compare the accuracy of the far-field and near-field models relative to each other in Section \ref{sec:validation}.

We now discuss a special case in which the proposed manifold reduces to manifold characterizations found in prior work \cite{FriedlanderExtendedManifoldAntennaArrays2020}. Let $\bmm_{n,k}^{(\ell)}$ be the moment of the $n$th antenna's $k$th segment when the $\ell$th antenna is excited by a unity current. The far-field model can be expanded as
\begin{align}
  \label{eq:8}
  \Earrff(\bp, \bw) = & \left[ \sum_{n=1}^N \sum_{k=1}^{K_n} \Tff(\bp_{n,k}) \bmm_{n,k}^{(1)}, \, \cdots, \right. \nonumber \\ & \left. \sum_{n=1}^N \sum_{k=1}^{K_n} \Tff(\bp_{n,k}) \bmm_{n,k}^{(N)}\right] \bw
\end{align}
Assume that for each antenna, the points $\bp_{n,k}$ for the $n$th antenna can be approximated by the single point $\bp_n$. Defining $\Mbar_{n,k} = \left[ \bmm_{n,k}^{(1)}, \, \cdots, \, \bmm_{n,k}^{(N)} \right]$,
\begin{align}
  \label{eq:4}
  \cbE_{\sf arr, ff}(\bp, \bw) = \sum_{n=1}^N  \Tff(\bp_n)  \sum_{k=1}^{K_n} \Mbar_{n,k} \bw.
\end{align}
If we further assume that the antennas in the array are not coupled to each other, then each antenna will only exhibit a current when it is directly excited. This implies $\Mbar_{n,k} = \left[ \mathbf{0}, \, \cdots, \, \bmm_{n,k}^{(n)}, \, \cdots, \, \mathbf{0}\right]$, which yields
\begin{align}
  \label{eq:22}
  & \cbE_{\sf arr, ff}(\bp, \bw)  = \sum_{n=1}^N  \Tff(\bp_n)  w_n \sum_{k=1}^{K_n} \bmm_{n,k}^{(n)} \\
                               & = \sum_{n=1}^N  \frac{\beta^2 e^{- \cj \beta \, \norm{\bp_n}}}{\cj \omega \epsilon_0 4 \pi \norm{\bp_n}}  \left[
        {\bf 0} \, \bua(\bp) \, \bue(\bp) \right] w_n \sum_{k=1}^{K_n} \bmm_{n,k}^{(n)}
\end{align}
where the second expression follows from the definition of $\Tff(\bp_n)$ in \eqref{eq:ff-ang-amp} and \eqref{eq:ff-dipole-transform}.
Defining the $n$th field pattern as
\begin{equation} 
  \label{eq:24}
  \bg(\bp_n) = \frac{\beta^2 }{\cj \omega \epsilon_0 4 \pi}  \left[
        {\bf 0} \, \bua(\bp) \, \bue(\bp) \right],
\end{equation}
the far-field electric field model reduces to
\begin{equation}
  \label{eq:23}
  \cbE_{\sf arr, ff}(\bp, \bw)  = \sum_{n=1}^N  \frac{e^{-\cj \beta \, \norm{\bp_n}}}{\norm{\bp_n}} \bg_n w_n.
\end{equation}
Careful inspection of this expression reveals that this model uses the \emph{isolated} array manifold in which the gain pattern of each antenna is unperturbed by the presence of the other antennas. Defining the isolated manifold as
\begin{equation}
  \label{eq:25}
  \bA_{\sf arr, iso}(\bp) = \left[\frac{e^{-\cj \beta \, \norm{\bp_0}}}{\norm{\bp_1}} \bg_1, \, \cdots, \, \frac{e^{-\cj \beta \, \norm{\bp_N}}}{\norm{\bp_N}} \bg_N   \right]
\end{equation}
gives
\begin{equation}
  \label{eq:E-field-isolated}
  \cbE_{\sf arr, ff}(\bp, \bw) = \bA_{\sf arr, iso}^{\T}(\bp) \bw.
\end{equation}
To summarize, the proposed model converges to the isolated manifold under the following conditions: each antenna is electrically small compared to distance $r$ and the antennas are uncoupled.

\section{Array model validation}
\label{sec:validation}

In this section, we validate the proposed modeling results by comparing our E-field approximations to measurements obtained from electromagnetic simulations. We simulated a four-element array of antennas using the MATLAB Antenna Toolbox. We test the model with two types of array: a homogeneous half-wave dipole array, and a heterogenous array composed of half-wave dipoles and V-dipoles. In both cases, the antennas are tuned to operate at an operating frequency of $5$ GHz. The MATLAB simulation is also used to obtain the effective array moment matrix $\bM$ by measuring the current distribution when different array elements are excited. MATLAB also automatically partitions the array into the segments as discussed in the previous section. The specifications for both types of antenna arrays, including the number segments per antenna, are given in Table \ref{tab:antenna-params} for completeness.

\begin{table}
  \centering
  \begin{tabular}{|r||l|l|}
    
 \hline
 Antenna type & Dipole & V-dipole \\
 \hline
  Operating freq. (GHz)  & 5   & 5 \\
  Length (mm) & 28.2 & 14.6 \\
  Width (mm) & 0.600 & 1.50 \\
  Flare angle $(^\circ)$ & n/a & 90 \\
  \# segments/ant & 40 & 100 \\
  \hline
\end{tabular}
\caption{\label{tab:antenna-params} Antenna specifications.}
\end{table} 

\subsection{Electric field}

In the first experiment, we obtain electric field measurements using MATLAB Antenna Toolbox and compare these against the proposed model. At a point $\bp$, we let $\cbE_{\sf sim}$ be the field measurements obtained from MATLAB and $\cbE_{\sf mod}$ be the field values obtained from the proposed model in \eqref{eq:E-lin-array}. As a benchmark, we also use the far-field approximation in which the rotational coherence is ignore, as discussed in Section II-C.

We calculate the relative error between the simulated and modeled fields as
\begin{equation}
  \label{eq:sim-error}
\text{Relative error} =  \frac{\norm{\cbE_{\sf sim} - \cbE_{\sf mod}}_{\sf F}}{\cbE_{\sf sim}}
\end{equation}
 as a function of distance from an eight-element dipole array for different element spacings in Fig. \ref{fig:E-field-validation}. We compare the proposed far-field and near-field models against two standard benchmarks: the isolated manifold and the embedded manifold. The isolated manifold uses the model in \eqref{eq:E-field-isolated} and uses the uncoupled field patterns of the antennas. The embedded manifold uses the same model but utilizes the coupled antenna far-field patterns utilized from MATLAB \cite{FriedlanderMutualCouplingMatrixArray2020}. We note that both of these baselines use far-field gain patterns. 

The results in Fig. \ref{fig:E-field-validation} show that the error for the proposed near-field model is relatively small at all distances. As expected, the far-field approximation converges to the near-field model as the distance increases. Comparing Figs. \ref{fig:E-field-validation}-(a) and  \ref{fig:E-field-validation}-(b), we see that larger element spacings increase the accuracy f the isolated manifold. This makes sense since larger element spacings will decrease mutual coupling and satisfy the assumptions needed to apply the isolated manifold model.

\begin{figure}
\centering
\begin{subfigure}{0.45\textwidth}
  \includegraphics[width=\textwidth]{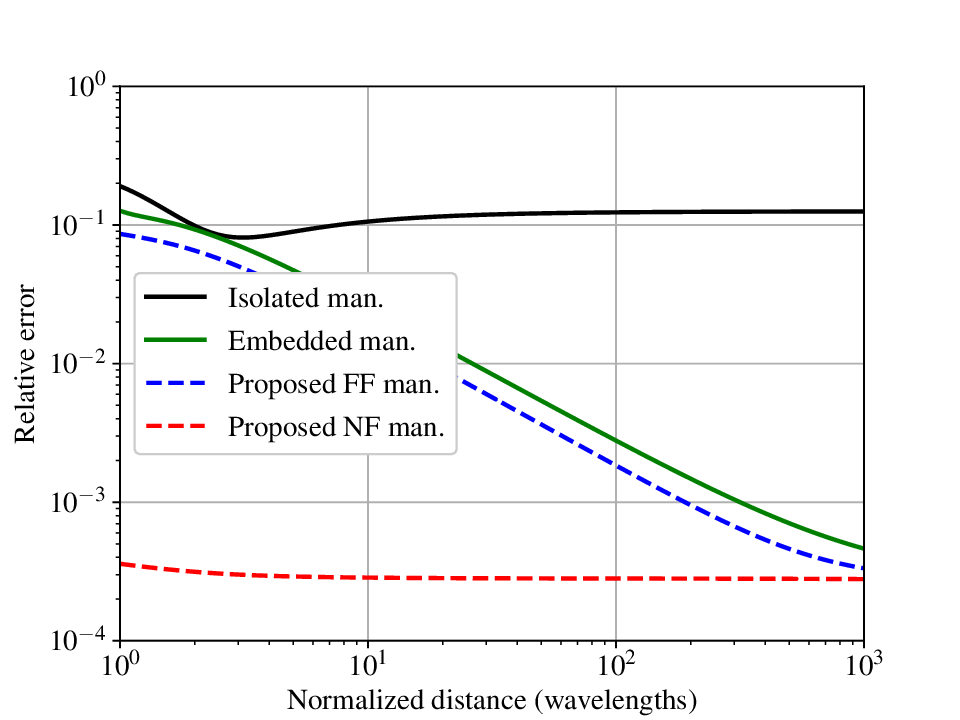}
    \caption{}
\end{subfigure}
\begin{subfigure}{0.45\textwidth}
    \includegraphics[width=\textwidth]{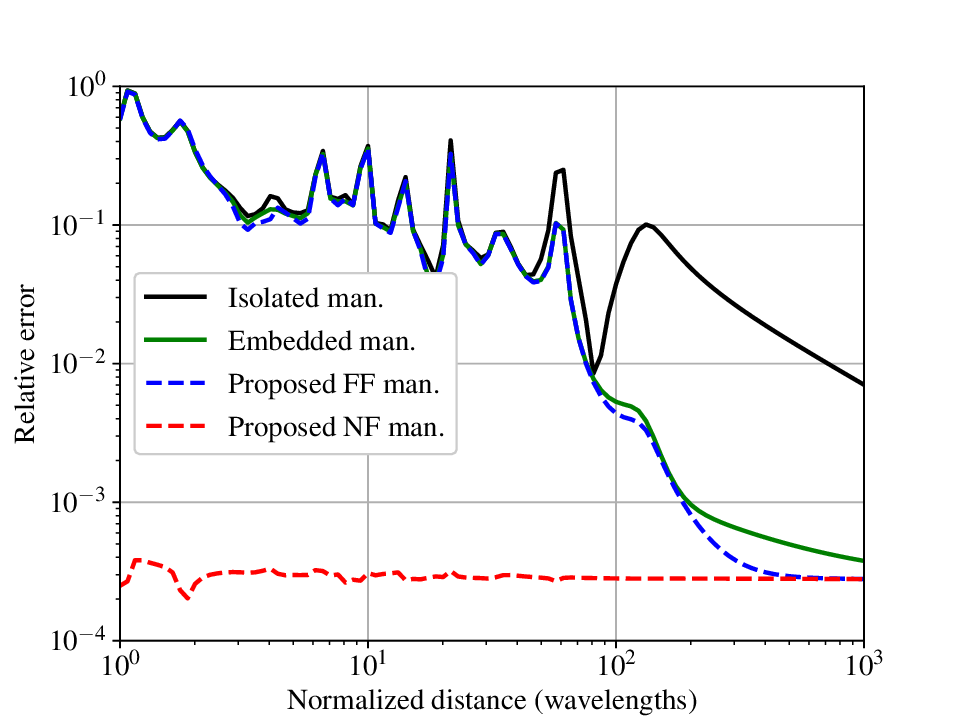}
    \caption{}
\end{subfigure} 
        
\caption{\label{fig:E-field-validation} Relative error between the electric field values obtained from electromagnetic simulations and different manifold models for an eight-element half-wavelength dipole array (a) with quarter-wavelength spacing and (b) four-wavelength spacing. The proposed model achieves lower error by accounting for near-field effects and polarization mismatches between antenna segments. At larger element spacings, the traditional manifold, which only uses the isolated far-field pattern, becomes more accurate due to decreased coupling.}
\end{figure}

In Fig. \ref{fig:E-field-validation-bowtie}, we show results for a similar experiment but replace a heterogenous array. The array composition is shown in Fig. \ref{fig:E-field-validation-bowtie}(a). As before, the proposed near-field error achieves a relatively low error compared to the MATLAB simulated values. In this case, the isolated manifold cannot be applied because both types of antennas have different field patterns. It is important to note that the embedded manifold exhibits a lower error, especially at larger distances. This is expected because the embedded patterns used from MATLAB use the same measurements as the simulated values. As the field pattern converges to its far-field characterization, the embedded manifold will become as accurate as the simulated value. Even in this case, the proposed manifold solution is preferable due to the large number of measurements that would be needed to characterize the embedded manifold. The embedded patterns must be measured separately for each antenna at all angles. In contrast, the proposed model uses the moment matrix $\Mbar$ to characterize both the near-field and far-field. The moments do not change as a function of position or distance, and so the proposed model offers a flexible alternative that is effective at all distances from the array.

\begin{figure}
\centering
\begin{subfigure}{0.45\textwidth}
  \includegraphics[width=\textwidth]{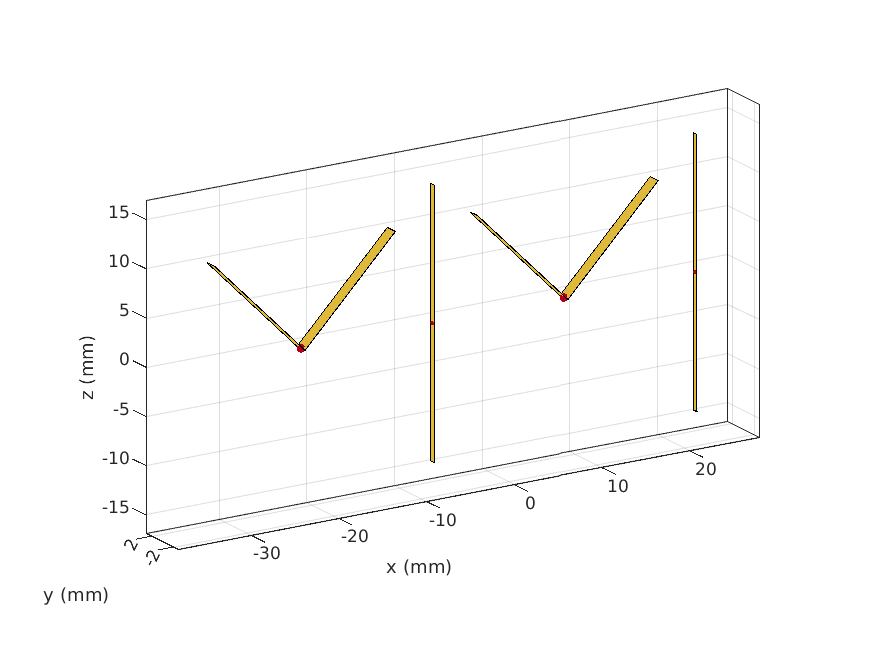}
    \caption{}
\end{subfigure}
\begin{subfigure}{0.45\textwidth} 
    \includegraphics[width=\textwidth]{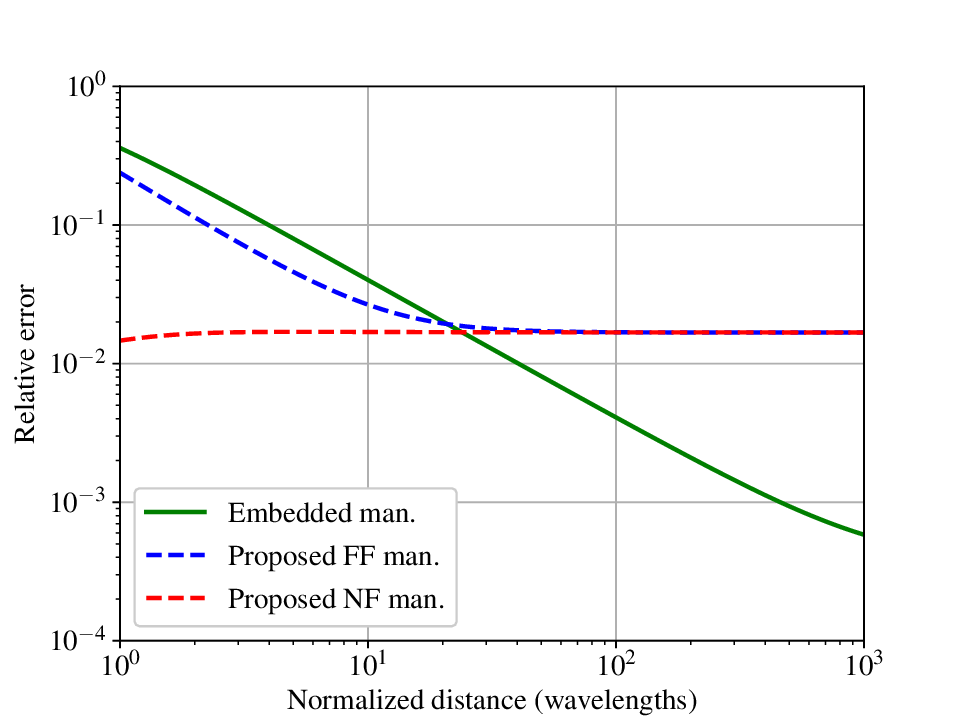}
    \caption{}
\end{subfigure} 
        
\caption{\label{fig:E-field-validation-bowtie} (a) Diagram of a heterogeneous array composed of half-wave dipoles and V-dipoles. (b) Relative error between the electric field values obtained from electromagnetic simulations and different manifold models. While the embedded manifold outperforms the proposed model at certain distances, it requires the far-field patterns of each antenna and is a more complicated characterization.}
\end{figure}

\subsection{Power density}

In the second validation experiment, we compute the plane-wave equivalent PD in the near-field of the antenna arrays. As before, we compare PD values from MATLAB and those obtained from the proposed model and the far-field benchmark. The PD value was computed as an average over 50 points uniformly distributed on a sphere of radius of two wavelengths. We further validate the model by comparing the PD for different beamforming angles. For an azimuth steering angle $\theta_{\sf st}$, we obtain the beamforming phases $\bm{\psi}(\theta_{\sf st})$ needed to steer the array in that particular direction using the MATLAB Antenna Toolbox. We then set $\bw = \text{exp}(-\cj \bm{\psi}(\theta_{\sf st}))$ in the electric field model to compute the resulting PD.

The PD validation results the two array types are displayed in Fig. \ref{fig:PD-validation}. The PD values are normalized by the maximum PD value obtained from simulations. In both cases, we note that the average PD over the sphere varies as a function of the steering angle. As before, the PD from the dipole array is well-approximated by the proposed model even at a distance of two wavelengths. In the bowtie case, the proposed model accurately predicts the PD, achieving near-perfect approximation for angles close to boresight.

  \begin{figure}
\centering
  \includegraphics[width=0.5\textwidth]{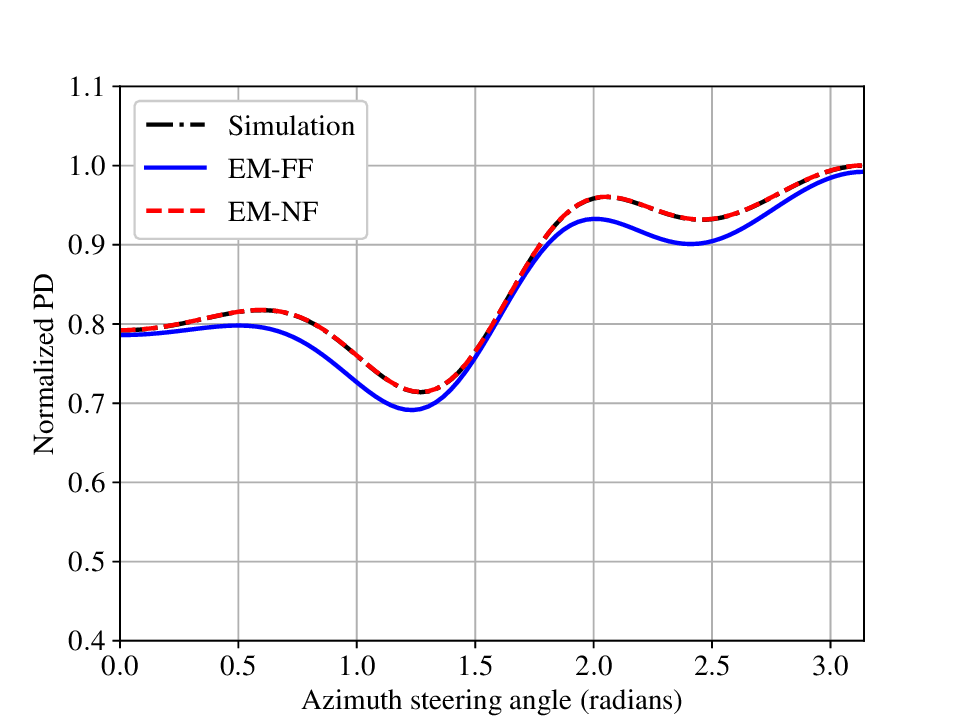}
        
\caption{\label{fig:PD-validation} PD values normalized by the maximum simulated value at different array steering angles for the dipole array. The proposed model approximates the PD well and can be used to estimate the radiated power.} 
\end{figure}

\section{Beamforming optimization}
\label{sec:beam-pattern}

We now apply the model developed in the previous section to formulate a beam pattern optimization problem under a variety of settings. We leverage the flexibility of the developed model to characterize near-field beam patterns and PD constraints. We focus mostly on the application of the model for directional beamforming, in which the weights are designed to maximize the sum of the field strength components at a particular spatial location.

\subsection{Maximizing field strength}

In the first approach, the transmitter designs $\bw$ to maximize the field strength at a point of interest to focus the signal. Imposing a power constraint of $P$ on the beamforming vector, the beam focusing problem can be formulated as
\begin{IEEEeqnarray}{rrl}
  \label{eq:ideal-opt}
  \bw_{\sf opt}(\bp) = & \, \, \argmax_{\bw} \, \, & \bw^* \Aarr^*(\bp) \Aarr(\bp) \bw, \\
  & \text{s.t.} \, \, & \norm{\bw}^2 \leq P.
\end{IEEEeqnarray}
The dependence of the solution on $\bp$ arises from the matrices $\Tarr(\bp)$ and $\Rarr(\bp)$. This is the well-known SVD beamforming solution as discussed in \cite{love2003grassmannian}. The key difference here is the use of the electromagnetic manifold, which accounts for polarization and can be leveraged at all distances from the array.

The weight vector maximizing the field amplitude can be found from the array model by noting the objective function is a Rayleigh quotient. Let $\Aarr(\bp)$ have SVD $\Aarr(\bp) = \bU \bS \bV^*$. Denoting the dominant right singular vector as $\bv_1$, the solution to \eqref{eq:ideal-opt} is $\bw_{\sf opt} =  \sqrt{P} \bv_1$. The resulting field strength is given by the dominant squared singular value $\sigma_1^2$.

\subsection{Maximizing field strength for certain polarization}

The optimization in the prior section ensures that the field strength is completely focused at a given point but ignores the polarization of the potential receiver. In this case, the array model can still be used to find the best beamforming vector. Let $\bb \in \C^3$ be a unit-norm vector that represents the receive polarization of an ideal omnidirectional antenna. A linearly polarized receiver with polarization angle $\psi$, for example, could be represented by the vector $\bb = [0,\cos \psi, \, \sin \psi]^{\T}$. The problem of finding the optimal weight vector maximize the field strength along the polarization $\bb$ is given by
\begin{IEEEeqnarray}{rrl}
  \label{eq:pol-weight-opt}
    \bw_{\bb, {\sf opt}}(\bp) = & \, \, \argmax_{ \bw} \, \, & \abs{ \bb^{*} \Aarr(\bp) \bw}^2, \\
  & \text{s.t.} \, \, & \norm{\bw}^2 \leq P.
\end{IEEEeqnarray}
Maximum ratio transmission (MRT) can be used to maximize the objective by simply setting $\bw_{\bb, {\sf opt}}(\bp) = \Aarr^*(\bp) \bb$. While this approach is simple, it reveals that the optimal weight vector is not constant as a function of the polarization $\bb$. One way to explain this result is by considering a diversely polarized array containing both vertically and horizontally polarized elements. Given a fixed transmit power, the choice of antennas that are excited will depend on whether the receiver is polarized horizontally, vertically, or along another direction.

The variation in the weight vector as a function of the receive polarization raises the notion of optimizing the receive polarization to maximize the field strength. In other words, the beam pattern synthesis problem becomes
\begin{IEEEeqnarray}{rrl}
  \label{eq:pol-weight-opt}
    (\bb_{\sf opt}(\bp), \, \bw_{{\sf opt}}(\bp)) = & \, \, \argmax_{\bw} \, \, & \abs{ \bb^{\T} \Aarr(\bp) \bw}^2, \\
  & \text{s.t.} \, \, & \norm{\bw}^2 \leq P.
\end{IEEEeqnarray}
Again, we let the SVD $\Aarr = \bU \bS \bV^*$. We also define $\bu_1$ and $\bv_1$ as the dominant left and right singular vectors. This gives the solutions $\bb_{\sf opt} = \bu_1^{\sf c}$ and $\bw_{\sf opt} = \bv_1$. As in the case of maximizing the field strength, the maximal objective value is simply the largest singular value of the decomposition.

The previous analysis reveals the following intuition of the SVD of $\Aarr = \bU \bS \bV^*$. The left singular matrix $\bU$ is associated with the \emph{inherent polarization basis} of the array. In other words, the columns of $\bU$ each represent a polarization of the radiated electric field, with the first column being the strongest polarization and the last column being the weakest. For example, a perfectly linearly polarized array would have a dominant left singular vector of $\bu_{\theta}$ or $\bu_{\phi}$. In general, we would expect the least dominant polarization to correspond with $\bu_{r}$, meaning that the field is weakest along the radial direction. The right singular matrix $\bV$ corresponds with the \emph{transmission modes} of the array. Beamforming along the manifold $\bA \bV$ represents exciting a particular polarization. The manifold dimensions dictate that the rank of $\Aarr$ is at most three, meaning there are only three modes that correspond with non-zero singular values. The manifold SVD shows how the modes can be used to synthesize any beam pattern as a linear combination of the array polarization basis.

\subsection{Beamforming with a PD constraint}

The results in Section \ref{sec:array-model} demonstrate that the plane-wave equivalent PD can be calculated in terms of the electromagnetic array manifold using \eqref{eq:pw-pd-model}. This expression can be leveraged to formulate a beam pattern synthesis problem with a near-field PD constraint. Exposure-constrained beamforming is highly relevant when designing beams for handsets and other devices that operate in close proximity to users.

We let the region $\Pcon$ denote the set of points over the PD will be constrained. For example, this could be defined as a sphere of a fixed radius around the device, or it could correspond to a surface modeling the skin of the user. The spatially-averaged PD, denoted as $\PD(\Pcon)$, can be found using \eqref{eq:pw-pd-model} as
\begin{equation}
  \label{eq:pd-avg}
  \PD(\Pcon) = \frac{1}{2 \eta_0 \abs{\Pcon}} \int_{\Pcon} \bw^* \bA^*(\bp) \bA(\bp) \bw \, d \bp
\end{equation}
We define a characteristic PD matrix $\bX_\Pcon$ for $\Pcon$ as the matrix that maps the antenna weights to the average PD over $\Pcon$. From \eqref{eq:pd-avg},
\begin{equation}
  \label{eq:char-pd-mat}
  \bX_\Pcon = \frac{1}{2 \eta_0 \abs{\Pcon}} \int_{\Pcon} \bA^*(\bp) \bA(\bp) \bw \, d \bp,
\end{equation}
which yields
\begin{equation}
  \label{eq:11}
  \PD(\Pcon) = \bw^* \bX_\Pcon \bw.
\end{equation}
As noted in prior work, this is a quadratic form in terms of the weight vector. A similar model for PD was discussed in \cite{castellanos2019signal}. There, however, the model did not include the effect of polarization and used correction terms to deal with mutual coupling and near-field effects.

We can leverage the quadratic PD model to formulate a PD-aware beam focusing problem. Let $Q$ be the regulatory constraint representing the maximum allowable PD in the region $\Pcon$. Then a PD constraint can be incorporated into the beam focusing problem in \eqref{eq:ideal-opt} as
\begin{IEEEeqnarray}{rrl}
  \label{eq:pol-pd-pow-weight-opt}
    \bw_{ {\sf opt}}(\bp) = & \, \, \argmax_{\bw} \, \, & \norm{\Aarr(\bp) \bw}^2,  \\
    & \text{s.t.} \, \, & \bw^* \bX \bw \leq Q, \nonumber \\
    & & \norm{\bw}^2 \leq P, \nonumber
\end{IEEEeqnarray}
where we have removed the dependence on $\Pcon$ for convenience. This problem reflects a realistic scenario for a user device in which the transmit power is limited to mitigate interference and the PD is regulated by exposure rules.

For simplicity of exposition, we will study a simpler problem in which only the exposure constraint is active. The solution to to the general problem in \label{eq:pol-pd-pow-weight-opt} can be obtained by exploiting the strong duality of the problem as shown in \cite{ying2015closed, BeckEldarStrongDualityNonconvexQuadratic2006}.
We have the problem
\begin{IEEEeqnarray}{rrl}
  \label{eq:pol-pd-weight-opt}
    \bw_{ {\sf opt}}(\bp) = & \, \, \argmax_{\bw} \, \, & \norm{ \Aarr(\bp) \bw}^2,  \\
    & \text{s.t.} \, \, & \bw^* \bX \bw \leq Q, \nonumber
  \end{IEEEeqnarray}
 The solution to this problem is given by the maximization of a generalized Raleigh quotient. Let $\Aarr \bX^{-1/2}$ have SVD $\Aarr \bX^{-1/2} = \bU_\bX \bS_\bX \bV_\bX^*$. Denoting the dominant right singular vector as $\bv_{\bX,1}$, the solution to \eqref{eq:pol-pd-weight-opt} is given by $\bw_{\sf opt} =  \sqrt{Q} \bX^{-1/2} \bv_{\bX,1}$.

\section{Beam pattern numerical results}
\label{sec:simulations}

In this section, we provide numerical results showing the performance of the proposed beamforming approaches in Section \ref{sec:beam-pattern}. The gain of the array is measured relative to a single half-wavelength dipole located at the origin. We use this metric because the half-wavelength dipole is both polarized and exhibits near-field effects. Let $\cbE_{\sf ref}(\bp)$ be the electric field radiated by a the reference half-wavelength dipole antenna. We then define the beamforming gain at $\bp$ as
\begin{equation}
  \label{eq:14}
  G(\bp) = \frac{\norm{\cbE_{\sf arr}(\bp)}^2}{\norm{\cbE_{\sf ref}(\bp)}^2}.
\end{equation}
When expressed in dB scale, gain relative to a dipole is generally given units dBd, and we use the same convention here. While the beamforming weights $\bw$ are optimized using the proposed manifold, the numerical gain values are obtained from electromagnetic simulations using the MATLAB Antenna Toolbox. In this manner, we show that the proposed beamforming approaches work in theory as well as in practice.

\subsection{Maximum field strength beamforming}

In the first simulation, we will use a 16 element ULA and compare the beamforming gain obtained from the proposed field strength optimization. Since the proposed optimization uses the developed electromagnetic manifold for beamforming, we use the isotropic manifold as a benchmark. Letting $\bc_n$ denote the position of the centroid of the $n$th antenna, the near-field isotropic manifold is defined as
\begin{equation}
  \ba_{\iso}(\bp) = \left[e^{-\cj \beta \Vert \bp - \bc_0 \Vert }, \dots, \, e^{-\cj \beta \Vert \bp - \bc_{N-1} \Vert }\right].
\end{equation}
The baseline algorithm uses the matched-filter weights $\bw_{\iso}(\bp) = \ba_{\iso}(\bp) / \norm{\ba_{\iso}(\bp)}$. Note that the isotropic manifold approach does not account for the polarization of the receiver without adjustments. We also compare the performance of the near-field and far-field manifold approximations discussed in Section III. We denote EM-FF as the far-field electromagnetic manifold far-field model, i.e., without the rotational coherence effect, and EM-NF as the near-field electromagnetic-manifold.

Fig. \ref{fig:dipole-beam} shows the beamforming gain of the antenna array that can be achieved by focusing using the proposed methods and the isotropic approach. In Fig. \ref{fig:dipole-beam}(a), we show the beamforming gain when the array is focused at a distance of 5 wavelengths away from the array center. All approaches show that the largest gain is achieved near the focusing distance. The electromagnetic approaches however, show a higher gain than the isotropic manifold. Interestingly, electromagnetic focusing also achieves lower gains than the isotropic approach at distances beyond the focusing distances. This behavior is akin to having lower sidelobes, which means the proposed approaches can provide both higher gain and reduced interference.

In Fig. \ref{fig:dipole-beam}(b), we show that maximum focusing gain that can be achieved at various focusing distances. In other words, each point on a curve corresponds to a point $\bp$ determined by the distance and direction. For each $\bp$, we obtain the weights $\bw$ according to \eqref{eq:ideal-opt} for the proposed approaches or according to the isotropic case. The results show that the proposed methods achieves a higher gain than isotropic beamforming.

\begin{figure}
\centering
\begin{subfigure}{0.45\textwidth}
  \includegraphics[width=\textwidth]{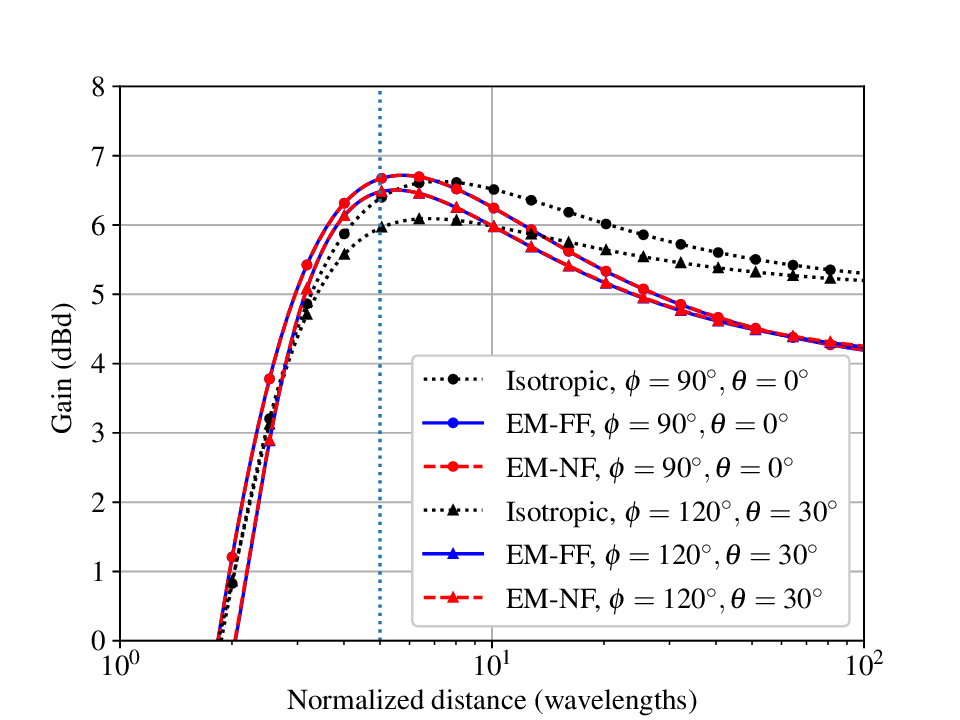}
    \caption{}
\end{subfigure}
\begin{subfigure}{0.45\textwidth}
    \includegraphics[width=\textwidth]{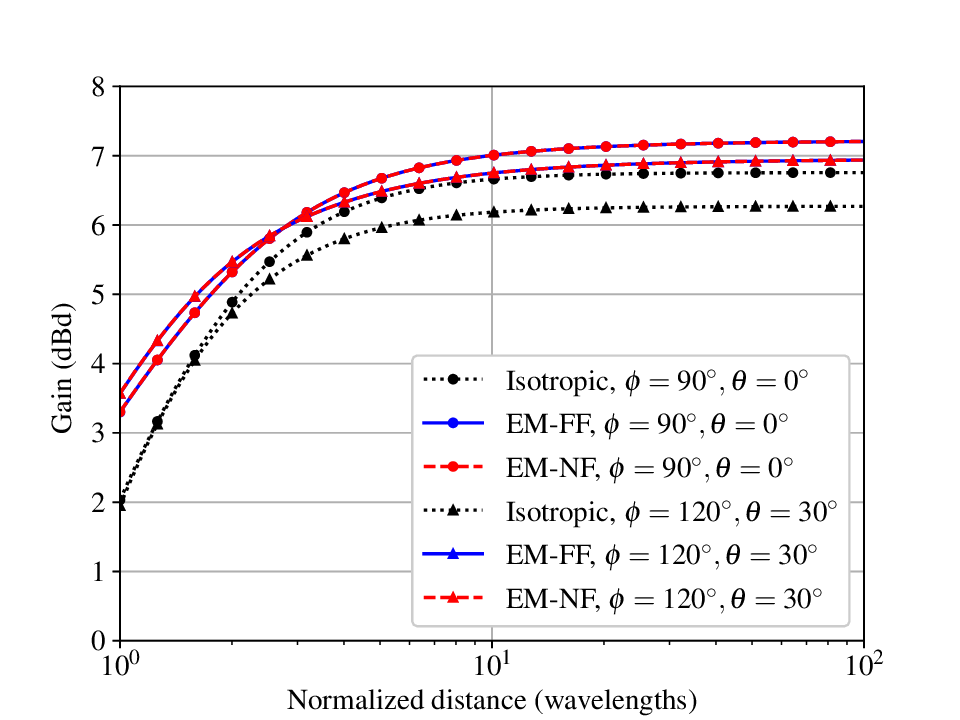}
    \caption{}
  \end{subfigure}
  \caption{\label{fig:dipole-beam} Gain of a 16-element half-wavelength dipole array v. distance. In (a), the beam is focused at a distance of 5 wavelengths away from the array. The focusing distance in (a) is shown by the dotted vertical line. The maximum gain that can be achieved through focusing throughout the distance range is shown in (b). Beamforming with the electromagnetic manifolds (EM-NF and EM-FF) shows increased gains over the isotropic approach.}
\end{figure}

\subsection{Effect of array spacing}

This experiment assesses how the distance between elements affects the efficacy of the proposed methods. We simulate a ULA array of 16 half-wavelength dipoles. One of the benefits of the proposed manifold method is that it inherently accounts for mutual coupling between array elements. The isotropic manifold can be adapted to account for mutual coupling by multiplying the length $N$ steering vector $\ba_{\iso}(\bp)$ by an $N \times N$ coupling matrix $\bC$. As noted in prior work \cite{FriedlanderMutualCouplingMatrixArray2020}, however, this approach neglects the fact that effect of mutual coupling actually depends on the point $\bp$. Here, we are more interested in how mutual coupling affects the gain of the proposed approach. For this reason, we use the unaltered isotropic manifold for comparison. 

The effect of array spacing on the array beamforming gain is show in Fig. \ref{fig:dipole-beam-spacing}. The $\sf x$-axis shows the normalized element spacing in wavelengths. At large element separations, we expect minimal coupling between the dipoles. Correspondingly, we see that the difference in gain between the electromagnetic approach and the isotropic approach is quite small in this regime. The array gain of the proposed method increases significantly over the benchmark as the array size grows over. Similar trends have been found in prior work \cite{ShyianovEtAlAchievableRateWithAntenna2022 ,IvrlacNossekTowardCircuitTheoryCommunication2010}. The results also show that the array gain trend is not monotonic. This behavior can be explained by the fact that the strength of mutual coupling waxes and wanes depending on the spacing.

\begin{figure}
  \centering %
  \includegraphics[width=0.5\textwidth]{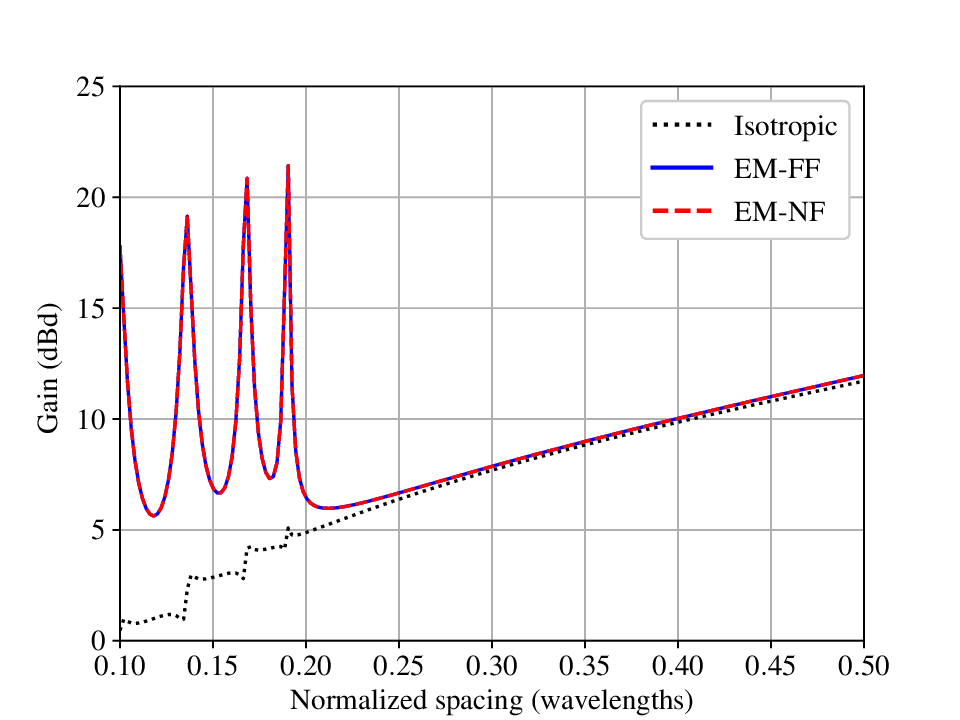}
  \caption{\label{fig:dipole-beam-spacing} Maximum focusing gain v. normalized spacing of a dipole array. The performance gap between the proposed methods and isotropic beamforming widens significantly as mutual coupling increases.}
\end{figure}

\subsection{PD-constrained beamforming}

In the final experiment, we apply a PD constraint a four-element dipole array spaced at quarter-wavelength. We simulated a scenario in which the transmitter beamforms at a distance of $100$ wavelengths in a particular direction in the far-field while under a power constraint. We consider two different kinds of constraints: a transmit power constraint and a radiated power constraint. The transmit power constraint is $\norm{\bw}^2 \leq P$, where $P$ is normalized by $1$ W. The radiated power constraint enforced a PD constraint over a uniformly sampled sphere as in the validation results. Here the radius of the sphere is taken to be one wavelength. To constrain the PD, we find $\bX$ according to the procedure in Section IV-C and normalize this matrix by the maximum average PD in the constrained region at a transmit power of 1 W. Doing so means that the constraint can be interpreted as a percentage of the maximum PD. To ensure that the PD constraint is met, we further check that this value lies below the constraint using MATLAB and reduce the transmit power as needed for compliance. In other words, if a beamformer $\bw$ achieves PD higher than the constraint, the transmitter finds an attenuation constant $\gamma$ such that $\gamma \bw$ satisfies the constraint and beamforms with the attenuated weights.

The numerical results in Fig. \ref{fig:dipole-beam-pd} compare the radiated field strength of the both the near-field and far-field models as a function of the power constraint. The proposed methods show that the radiated power constraint results in higher beamforming gains. This is because the radiated power constraint generally results in higher power beams compared to the transmit power constraint. In addition, the results show that the proposed methods achieve a significant gain over isotropic beamforming. They also show that the far-field model achieves lower gain than the near-field approach. This is because the far-field model leads to errors in calculating the power density as shown in the validation results. Those errors mean that the beamformer found from solving \eqref{eq:pol-pd-weight-opt} may not be optimal or feasible. In the case that it is not feasible, power back-off further deteriorates the performance. This demonstrates the importance of accurate PD models; better radiated power approximations lead to higher gains.

\begin{figure}
  \centering %
  \includegraphics[width=0.5\textwidth]{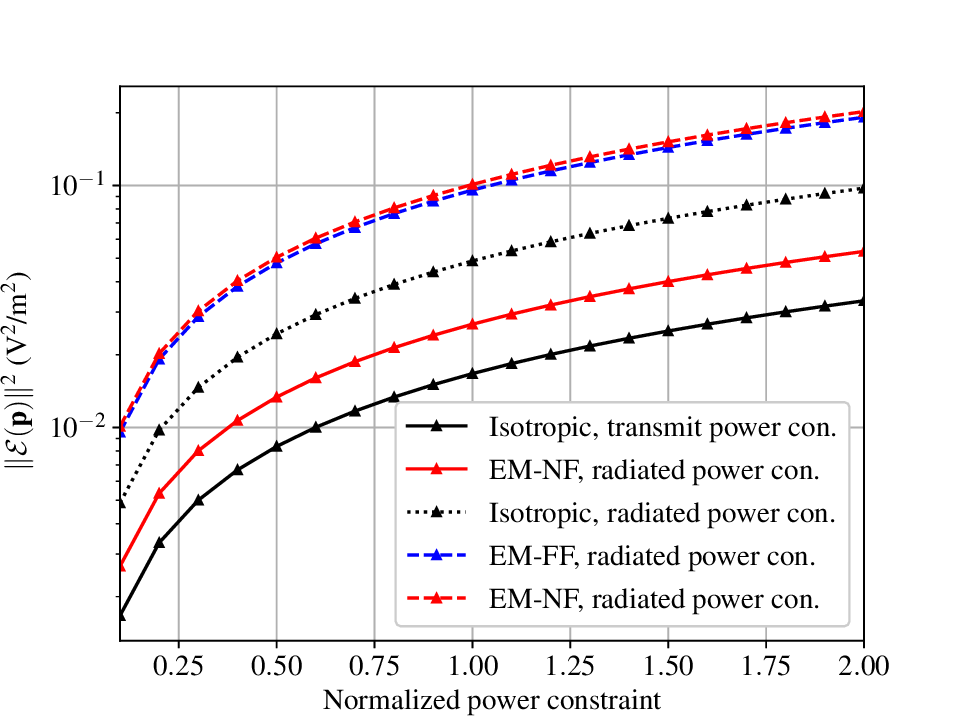
  }
  \caption{\label{fig:dipole-beam-pd} Electric field strength of a four-element half-wave dipole array under both transmit power and radiated power constraints. The radiated power constraint results in higher beamforming gains since the transmit power can be increased without violating the radiated power limit.}
\end{figure} 

\section{Conclusion}
\label{sec:conclusion}

The electromagnetic-based manifold model presented in this paper can be leveraged for a variety of important applications. By modeling arbitrary antenna arrays as a collection of Hertzian dipoles, we have developed a mathematically tractable manifold that only depends on the current distribution over the array. We showed how the manifold can be used to approximate electromagnetic fields, calculate PD, and perform near-field beamforming. The proposed approach accounts for a number of important interactions including mutual coupling, near-field propagation and polarization. We stress that the proposed model simultaneously provides a physically-consistent approach for antenna modeling and retains enough structure for the application of wireless and array processing theory. Our future research will focus on extending the results to polarized MIMO systems with heterogenous array at both ends.

\bibliographystyle{IEEEtran}
\bibliography{refs_rc,exp_refs}

\begin{thebibliography}{10}
\providecommand{\url}[1]{#1}
\csname url@samestyle\endcsname
\providecommand{\newblock}{\relax}
\providecommand{\bibinfo}[2]{#2}
\providecommand{\BIBentrySTDinterwordspacing}{\spaceskip=0pt\relax}
\providecommand{\BIBentryALTinterwordstretchfactor}{4}
\providecommand{\BIBentryALTinterwordspacing}{\spaceskip=\fontdimen2\font plus
\BIBentryALTinterwordstretchfactor\fontdimen3\font minus
  \fontdimen4\font\relax}
\providecommand{\BIBforeignlanguage}[2]{{%
\expandafter\ifx\csname l@#1\endcsname\relax
\typeout{** WARNING: IEEEtran.bst: No hyphenation pattern has been}%
\typeout{** loaded for the language `#1'. Using the pattern for}%
\typeout{** the default language instead.}%
\else
\language=\csname l@#1\endcsname
\fi
#2}}
\providecommand{\BIBdecl}{\relax}
\BIBdecl

\bibitem{BjoernsonEtAlMassiveMimoIsReality2019}
E.~Bj{\"o}rnson, L.~Sanguinetti, H.~Wymeersch, J.~Hoydis, and T.~L. Marzetta,
  ``Massive {MIMO} is a reality—{What} is next?: {Five}q promising research
  directions for antenna arrays,'' \emph{Digital Signal Processing}, vol.~94,
  pp. 3--20, Nov. 2019.

\bibitem{YuanEtAlUltraWidebandMimoAntenna2020}
X.-T. Yuan, W.~He, K.-D. Hong, C.-Z. Han, Z.~Chen, and T.~Yuan,
  ``Ultra-wideband {MIMO} antenna system with high element-isolation for {5G}
  smartphone application,'' \emph{IEEE Access}, vol.~8, pp. 56\,281--56\,289,
  Mar. 2020.

\bibitem{ZhangEtAlUltraWideband8Port2019}
X.~Zhang, Y.~Li, W.~Wang, and W.~Shen, ``Ultra-wideband 8-port {MIMO} antenna
  array for {5G} metal-frame smartphones,'' \emph{IEEE Access}, vol.~7, pp.
  72\,273--72\,282, May 2019.

\bibitem{ShlezingerEtAlDynamicMetasurfaceAntennas6g2021}
N.~Shlezinger, G.~C. Alexandropoulos, M.~F. Imani, Y.~C. Eldar, and D.~R.
  Smith, ``Dynamic metasurface antennas for {6G} extreme massive {MIMO}
  communications,'' \emph{IEEE Trans. Wireless Commun.}, vol.~28, no.~2, pp.
  106--113, Apr. 2021.

\bibitem{BhaMaDic:RESHAPE:-A-Liquid-Metal-Based:21}
V.~T. Bharambe, J.~Ma, M.~D. Dickey, and J.~J. Adams, ``{RESHAPE}: A liquid
  metal-based reshapable aperture for compound frequency, pattern, and
  polarization reconfiguration,'' \emph{IEEE Trans. Antennas Propag.}, vol.~69,
  no.~5, pp. 2581--2594, May 2021.

\bibitem{ClerckxEtAlImpactAntennaCoupling2007}
B.~Clerckx, C.~Craeye, D.~Vanhoenacker-Janvier, and C.~Oestges, ``Impact of
  antenna coupling on 2 $\times$ 2 {MIMO} communications,'' \emph{IEEE Trans.
  Veh. Technol.}, vol.~56, no.~3, pp. 1009--1018, May 2007.

\bibitem{OzdemirEtAlDynamicsSpatialCorrelationAnd2004}
M.~Ozdemir, E.~Arvas, and H.~Arslan, ``Dynamics of spatial correlation and
  implications on {MIMO} systems,'' \emph{IEEE Commun. Mag.}, vol.~42, no.~6,
  pp. S14--S19, Jun. 2004.

\bibitem{IvrlacEtAlFadingCorrelationsWirelessMimo2003}
M.~Ivrlac, W.~Utschick, and J.~Nossek, ``Fading correlations in wireless {MIMO}
  communication systems,'' \emph{IEEE J. Sel. Areas Commun.}, vol.~21, no.~5,
  pp. 819--828, Jun. 2003.

\bibitem{WallaceJensenModelingIndoorMimoWireless2002}
J.~Wallace and M.~Jensen, ``Modeling the indoor {MIMO} wireless channel,''
  \emph{IEEE Trans. Antennas Propag.}, vol.~50, no.~5, pp. 591--599, May 2002.

\bibitem{HeathLozanoFoundationsMIMOCommunications2018}
R.~W. Heath and A.~Lozano, \emph{Foundations of {{MIMO Communications}}}.\hskip
  1em plus 0.5em minus 0.4em\relax {Cambridge University Press}, 2018.

\bibitem{LeeEtAl28GhzMillimeterWave2018}
J.-H. Lee, J.-S. Choi, J.-Y. Lee, and S.-C. Kim, ``28 {GHz} millimeter-wave
  channel models in urban microcell environment using three-dimensional ray
  tracing,'' \emph{IEEE Antennas Wireless Propag. Lett.}, vol.~17, no.~3, pp.
  426--429, Mar. 2018.

\bibitem{NgEtAlEfficientMultielementRayTracing2007}
K.~H. Ng, E.~K. Tameh, A.~Doufexi, M.~Hunukumbure, and A.~R. Nix, ``Efficient
  multielement ray tracing with site-specific comparisons using measured {MIMO}
  channel data,'' \emph{IEEE Trans. Veh. Technol.}, vol.~56, no.~3, pp.
  1019--1032, May 2007.

\bibitem{BhaOesHea:A-New-Double-Directional-Channel-Model:10}
R.~{Bhagavatula}, C.~{Oestges}, and R.~W. {Heath}, ``A new double-directional
  channel model including antenna patterns, array orientation, and
  depolarization,'' \emph{IEEE Trans. Veh. Technol.}, vol.~59, no.~5, pp.
  2219--2231, Jun. 2010.

\bibitem{FriedlanderMutualCouplingMatrixArray2020}
B.~Friedlander, ``On the mutual coupling matrix in array signal processing,''
  in \emph{Proc. 54th Asilomar Conf. Signals, Syst., Comput.}, Pacific Grove,
  CA, USA, Nov. 2020, pp. 1245--1249.

\bibitem{DardariCommunicatingWithLargeIntelligent2020}
D.~Dardari, ``Communicating with large intelligent surfaces: Fundamental limits
  and models,'' \emph{IEEE J. Sel. Areas Commun.}, vol.~38, no.~11, pp.
  2526--2537, Nov. 2020.

\bibitem{BjornsonSanguinettiPowerScalingLawsAnd2020}
E.~Bj{\"o}rnson and L.~Sanguinetti, ``Power scaling laws and near-field
  behaviors of massive {MIMO} and intelligent reflecting surfaces,'' \emph{IEEE
  Open J. Commun. Soc.}, vol.~1, pp. 1306--1324, Sep. 2020.

\bibitem{PizzoEtAlFourierPlaneWaveSeries2022}
A.~Pizzo, L.~Sanguinetti, and T.~L. Marzetta, ``Fourier plane-wave series
  expansion for holographic {MIMO} communications,'' \emph{IEEE Trans. Wireless
  Commun.}, vol.~21, no.~9, pp. 6890--6905, Sep. 2022.

\bibitem{PizzoEtAlSpatiallyStationaryModelHolographic2020}
A.~Pizzo, T.~L. Marzetta, and L.~Sanguinetti, ``Spatially-stationary model for
  holographic {MIMO} small-scale fading,'' \emph{IEEE J. Sel. Areas Commun.},
  vol.~38, no.~9, pp. 1964--1979, Sep. 2020.

\bibitem{FriedlanderExtendedManifoldAntennaArrays2020}
B.~Friedlander, ``The extended manifold for antenna arrays,'' \emph{IEEE Trans.
  Signal Process.}, vol.~68, pp. 493--502, Jan. 2020.

\bibitem{FriedlanderAntennaArrayManifoldsHigh2018}
B.~Friedlander, ``Antenna array manifolds for high-resolution direction
  finding,'' \emph{IEEE Trans. Signal Process.}, vol.~66, no.~4, pp. 923--932,
  Feb. 2018.

\bibitem{FriedlanderPolarizationSensitivityAntennaArrays2019}
B.~Friedlander, ``Polarization sensitivity of antenna arrays,'' \emph{IEEE
  Trans. Signal Process.}, vol.~67, no.~1, pp. 234--244, Jan. 2019.

\bibitem{FriedlanderLocalizationSignalsNearField2019}
B.~Friedlander, ``Localization of signals in the near-field of an antenna
  array,'' \emph{IEEE Trans. Signal Process.}, vol.~67, no.~15, pp. 3885--3893,
  Aug. 2019.

\bibitem{LebretBoydAntennaArrayPatternSynthesis1997}
H.~Lebret and S.~Boyd, ``Antenna array pattern synthesis via convex
  optimization,'' \emph{IEEE Trans. Signal. Process.}, vol.~45, no.~3, pp.
  526--532, Mar. 1997.

\bibitem{CaoEtAlConstantModulusShapedBeam2017}
P.~Cao, J.~S. Thompson, and H.~Haas, ``Constant modulus shaped beam synthesis
  via convex relaxation,'' \emph{IEEE Antennas Wireless Propag. Lett.},
  vol.~16, pp. 617--620, Jul. 2017.

\bibitem{AlkhateebEtAlChannelEstimationAndHybrid2014}
A.~Alkhateeb, O.~El~Ayach, G.~Leus, and R.~W. Heath, ``Channel estimation and
  hybrid precoding for millimeter wave cellular systems,'' \emph{IEEE J. Sel.
  Topics Signal Process.}, vol.~8, no.~5, pp. 831--846, Oct. 2014.

\bibitem{LoveEtAlGrassmannianBeamformingMultipleInput2003}
D.~Love, R.~Heath, and T.~Strohmer, ``Grassmannian beamforming for
  multiple-input multiple-output wireless systems,'' \emph{IEEE Trans. Inf.
  Theory}, vol.~49, no.~10, pp. 2735--2747, 2003.

\bibitem{RaghavanEtAlSystematicCodebookDesignsQuantized2007}
V.~Raghavan, R.~W. Heath, and A.~M. Sayeed, ``Systematic codebook designs for
  quantized beamforming in correlated {MIMO} channels,'' \emph{IEEE J. Sel.
  Areas Commun.}, vol.~25, no.~7, pp. 1298--1310, Sep. 2007.

\bibitem{GemechuEtAlBeampatternSynthesisWithSidelobe2020}
A.~Y. Gemechu, G.~Cui, X.~Yu, and L.~Kong, ``Beampattern synthesis with
  sidelobe control and applications,'' \emph{IEEE Trans. Antennas Propag.},
  vol.~68, no.~1, pp. 297--310, Sep. 2020.

\bibitem{CuiEtAlEffectiveArtificialNeuralNetwork2021}
C.~Cui, W.~T. Li, X.~T. Ye, P.~Rocca, Y.~Q. Hei, and X.~W. Shi, ``An effective
  artificial neural network-based method for linear array beampattern
  synthesis,'' \emph{IEEE Trans. Antennas Propag.}, vol.~69, no.~10, pp.
  6431--6443, Apr. 2021.

\bibitem{ZhangSerRobustBeampatternSynthesisAntenna2011}
T.~Zhang and W.~Ser, ``Robust beampattern synthesis for antenna arrays with
  mutual coupling effect,'' \emph{IEEE Trans. Antennas Propag.}, vol.~59,
  no.~8, pp. 2889--2895, Aug. 2011.

\bibitem{SchmidEtAlEffectsCalibrationErrorsAnd2013}
C.~M. Schmid, S.~Schuster, R.~Feger, and A.~Stelzer, ``On the effects of
  calibration errors and mutual coupling on the beam pattern of an antenna
  array,'' \emph{IEEE Trans. Antennas Propag.}, vol.~61, no.~8, pp. 4063--4072,
  Aug. 2013.

\bibitem{XiaoNehoraiOptimalPolarizedBeampatternSynthesis2009}
J.-J. Xiao and A.~Nehorai, ``Optimal polarized beampattern synthesis using a
  vector antenna array,'' \emph{IEEE Trans. Signal. Process.}, vol.~57, no.~2,
  pp. 576--587, Feb. 2009.

\bibitem{FuchsFuchsOptimalPolarizationSynthesisArbitrary2011}
B.~Fuchs and J.~J. Fuchs, ``Optimal polarization synthesis of arbitrary arrays
  with focused power pattern,'' \emph{IEEE Trans. Antennas Propag.}, vol.~59,
  no.~12, pp. 4512--4519, Dec. 2011.

\bibitem{MyersEtAlNearFieldFocusingUsing2022}
N.~J. Myers, Y.~Aslan, and G.~Joseph, ``Near-field focusing using phased arrays
  with dynamic polarization control,'' in \emph{Proc. 30th Eur. Signal Process.
  Conf. (EUSIPCO)}, Belgrade, Serbia, Aug. 2022, pp. 1831--1835.

\bibitem{CastellanosHeathLinearPolarizationOptimizationWideband2023}
M.~R. Castellanos and R.~W. Heath, ``Linear polarization optimization for
  wideband {MIMO} systems with reconfigurable arrays,'' \emph{\emph{Accepted
  in} IEEE Trans. Wireless Commun.}, Jul. 2023.

\bibitem{ShyianovEtAlAchievableRateWithAntenna2022}
V.~Shyianov, M.~Akrout, F.~Bellili, A.~Mezghani, and R.~W. Heath, ``Achievable
  rate with antenna size constraint: {Shannon} meets {Chu} and {Bode},''
  \emph{IEEE Trans. Commun.}, vol.~70, no.~3, pp. 2010--2024, Mar. 2022.

\bibitem{JensenWallaceCapacityContinuousSpaceElectromagnetic2008}
M.~A. Jensen and J.~W. Wallace, ``Capacity of the continuous-space
  electromagnetic channel,'' \emph{IEEE Trans. Antennas Propag.}, vol.~56,
  no.~2, pp. 524--531, Feb. 2008.

\bibitem{ying2015closed}
D.~Ying, D.~J. Love, and B.~M. Hochwald, ``Closed-loop precoding and capacity
  analysis for multiple antenna wireless systems with user radiation exposure
  constraints,'' \emph{{IEEE} Trans. Wireless Commun.}, vol.~14, no.~10, pp.
  5859--5870, Oct. 2015.

\bibitem{hochwald2013sar}
B.~M. Hochwald, D.~J. Love, S.~Yan, and J.~Jin, ``{SAR} codes,'' in \emph{UCSD
  Information Theory Appl. Workshop (ITA)}, San Diego, CA, USA, Feb. 2013, pp.
  1--9.

\bibitem{castellanos2021dynamic}
M.~R. Castellanos, D.~Ying, D.~J. Love, B.~Peleato, and B.~M. Hochwald,
  ``Dynamic electromagnetic exposure allocation for {Rayleigh} fading {MIMO}
  channels,'' \emph{IEEE Trans. Wireless Commun.}, vol.~20, no.~2, pp.
  728--740, Feb. 2021.

\bibitem{ying2017sum}
D.~Ying, D.~J. Love, and B.~M. Hochwald, ``Sum-rate analysis for multi-user
  {MIMO} systems with user exposure constraints,'' \emph{{IEEE} Trans. Wireless
  Commun.}, vol.~16, no.~11, pp. 7376--7388, Sep. 2017.

\bibitem{castellanos2016hybrid}
M.~R. Castellanos, D.~J. Love, and B.~M. Hochwald, ``Hybrid precoding for
  millimeter wave systems with a constraint on user electromagnetic radiation
  exposure,'' in \emph{Proc. IEEE Asilomar Conf. Signals, Systems, and
  Computers}, Pacific Grove, CA, USA, Nov. 2016, pp. 296--300.

\bibitem{ying2013beamformer}
D.~Ying, D.~J. Love, and B.~M. Hochwald, ``Beamforming optimization with a
  constraint on user electromagnetic radiation exposure,'' in \emph{Proc. Conf.
  Information Sciences and Systems ({CISS})}, Baltimore, MD, USA, Mar. 2013.

\bibitem{HelJamBro:Exposure-Modelling-and-Minimization:20}
F.~H{\'e}liot, M.~A. Jamshed, and T.~W.~C. Brown, ``Exposure modelling and
  minimization for multi-antenna communication systems,'' in \emph{Proc. IEEE
  91st Veh. Technol. Conf. (VTC Spring)}, May 2020, pp. 1--6.

\bibitem{ebadi2018determining}
A.~{Ebadi-Shahrivar}, P.~{Fay}, D.~J. {Love}, and B.~M. {Hochwald},
  ``Determining electromagnetic exposure compliance of multi-antenna devices in
  linear time,'' \emph{{IEEE} Trans. Antennas Propag.}, vol.~67, no.~12, pp.
  7585--7596, Dec. 2019.

\bibitem{li2017high}
J.~Li, S.~Yan, Y.~Liu, B.~M. Hochwald, and J.-M. Jin, ``A high-order model for
  fast estimation of electromagnetic absorption induced by multiple
  transmitters in portable devices,'' \emph{{IEEE} Trans. Antennas Propag.},
  vol.~65, no.~12, pp. 6768--6778, Dec. 2017.

\bibitem{hochwald2014incorporating}
B.~M. Hochwald, D.~J. Love, S.~Yan, P.~Fay, and J.-M. Jin, ``Incorporating
  specific absorption rate constraints into wireless signal design,''
  \emph{{IEEE} Commun. Mag.}, vol.~52, no.~9, pp. 126--133, Sep. 2014.

\bibitem{castellanos2019signal}
M.~R. Castellanos, Y.~Liu, D.~J. Love, B.~Peleato, J.-M. Jin, and B.~M.
  Hochwald, ``Signal-level models of pointwise electromagnetic exposure for
  millimeter wave communication,'' \emph{{IEEE} Trans. Antennas Propag.},
  vol.~68, no.~5, pp. 3963--3977, May 2020.

\bibitem{StutzmanThieleAntennaTheoryAndDesign2012}
W.~L. Stutzman and G.~A. Thiele, \emph{Antenna theory and design}.\hskip 1em
  plus 0.5em minus 0.4em\relax John Wiley \& Sons, 2012.

\bibitem{lin2020fcc}
J.~C. {Lin}, ``{FCC} announces its existing {RF} exposure limits apply to
  {5G},'' \emph{{IEEE} Microw. Mag.}, vol.~21, no.~4, pp. 15--17, Apr. 2020.

\bibitem{love2003grassmannian}
D.~J. Love, R.~W. Heath, and T.~Strohmer, ``Grassmannian beamforming for
  multiple-input multiple-output wireless systems,'' \emph{IEEE transactions on
  information theory}, vol.~49, no.~10, pp. 2735--2747, 2003.

\bibitem{BeckEldarStrongDualityNonconvexQuadratic2006}
A.~Beck and Y.~C. Eldar, ``Strong duality in nonconvex quadratic optimization
  with two quadratic constraints,'' \emph{SIAM J. Optim.}, vol.~17, no.~3, pp.
  844--860, 2006.

\bibitem{IvrlacNossekTowardCircuitTheoryCommunication2010}
M.~T. Ivrla{\v c} and J.~A. Nossek, ``Toward a circuit theory of
  communication,'' \emph{{IEEE} Trans. Circuits Syst. {I}}, vol.~57, no.~7, pp.
  1663--1683, Jul. 2010.

\end{thebibliography}

\end{document}